\documentclass[11 pt]{article}

\usepackage[section]{placeins}
\usepackage{graphicx,subfigure}
\usepackage{amsmath,amssymb}
\usepackage{cite}
\usepackage{url}
\usepackage{color}

\newcommand{\dff}{\stackrel{\scriptscriptstyle\triangle}{=}}

\newcommand{\bbe}{\mathbb{E}}
\newcommand{\sinr}{\footnotesize \mbox{\sf SINR}}

\newcommand{\snr}{\footnotesize  \mbox{\sf SNR}}
\newcommand{\tsinr}{{\sf sinr}}
\newcommand{\tslinr}{{\sf slinr}}
\newcommand{\bsinr}{\overline{\sf sin}{\sf r}}
\newcommand{\mse}{\mbox{\sf MSE}}
\newcommand{\tmse}{{\sf M}\widetilde{\sf S}{\sf E}}

\newcommand{\ba}{\boldsymbol{a}}
\newcommand{\bb}{\boldsymbol{b}}
\newcommand{\bh}{\boldsymbol{h}}

\newcommand{\bDelta}{\boldsymbol{\Delta}}
\newcommand{\bx}{\boldsymbol{x}}
\newcommand{\bI}{\boldsymbol{I}}
\newcommand{\bV}{\boldsymbol{V}}
\newcommand{\bw}{\boldsymbol{w}}
\newcommand{\bW}{\boldsymbol{W}}

\newcommand{\bP}{\boldsymbol{P}}
\newcommand{\bd}{\boldsymbol{d}}
\newcommand{\bt}{\boldsymbol{t}}
\newcommand{\bc}{\boldsymbol{c}}
\newcommand{\be}{\boldsymbol{e}}
\newcommand{\bu}{\boldsymbol{u}}
\newcommand{\bU}{\boldsymbol{U}}

\newcommand{\bB}{\boldsymbol{B}}
\newcommand{\bT}{\boldsymbol{T}}
\newcommand{\bA}{\boldsymbol{A}}
\newcommand{\bC}{\boldsymbol{C}}
\newcommand{\bD}{\boldsymbol{D}}
\newcommand{\bQ}{\boldsymbol{Q}}
\newcommand{\blambda}{\boldsymbol{\lambda}}

\newcommand{\bbf}{\boldsymbol{f}}
\newcommand{\bPhi}{\boldsymbol{\Phi}}
\newcommand{\bPsi}{\boldsymbol{\Psi}}
\newcommand{\bbeta}{\boldsymbol{\beta}}

\newcommand{\R}{\mathcal{R}}
\newcommand{\CS}{\mathcal{S}}
\newcommand{\CP}{\mathcal{P}}

\newtheorem{theorem}{Theorem}
\newtheorem{lemma}{Lemma}
\newtheorem{coro}{Corollary}

\def\QED{\mbox{\rule[0pt]{1.5ex}{1.5ex}}}
\def\proof{\noindent {\it Proof: }}
\def\endproof{\hspace*{\fill}~\QED\par\endtrivlist\unskip}

\begin{document}

\title{Robust Linear Precoder Design for Multi-cell \\ Downlink Transmission}

\author{Ali Tajer\footnote{Electrical Engineering Department, Princeton
University, Princeton, NJ 08544.}\and Narayan Prasad\footnote{NEC Labs
America, Princeton, NJ 08540.}\and Xiaodong Wang\footnote{Electrical Engineering Department, Columbia University, New York, NY 10027.}}
\date{}

\maketitle

\allowdisplaybreaks

\begin{abstract}
Coordinated information processing by the base stations of  {\em multi-cell} wireless networks enhances the overall quality of communication in the network. Such coordinations for optimizing any desired network-wide quality of service (QoS) necessitate the base stations to acquire and share some channel state information (CSI). With perfect knowledge of channel states, the base stations can adjust their transmissions for achieving a network-wise QoS optimality. In practice, however, the CSI can be obtained only {\em imperfectly}. As a result, due to the uncertainties involved, the network is not guaranteed to benefit from a globally optimal QoS. Nevertheless, if the channel estimation perturbations are confined within bounded regions, the QoS measure will also lie within a bounded region. Therefore, by exploiting the notion of {\em robustness} in the {\em worst-case} sense some worst-case QoS guarantees for the network can be asserted. We adopt a popular model for noisy channel estimates that assumes that estimation noise terms lie within known hyper-spheres. We aim to design linear transceivers that optimize a {\em worst-case} QoS measure in downlink transmissions. In  particular, we focus on maximizing the {\em worst-case weighted sum-rate} of the network and the  {\em minimum worst-case rate} of the network. For obtaining such transceiver designs, we offer several centralized (fully cooperative) and distributed (limited cooperation) algorithms which entail different levels of complexity and information exchange among the base stations.
\end{abstract}

\section{Introduction}
\label{sec:intro}

The increasing demand for accommodating more users within wireless networks makes such networks interference limited. Multi-cell multiuser networks suffer from intra-cell interference as a consequence of having the base stations serve multiple users simultaneously, as well as from inter-cell interference among neighboring cells due to their ever shrinking sizes. A useful approach for mitigating the interference in downlink transmissions is to equip the base stations with multiple transmit antennas and employ transmit precoding. Such precoding exploits the spatial dimension to ensure that the signals intended for different users remain easily separable at their designated receivers. For enabling precoded transmission, the base stations should acquire the knowledge of channel states or channel state information (CSI).

The problem of designing linear precoders for {\em single}-cell downlink transmissions has been extensively investigated in the literature. Assuming {\em perfect knowledge of channel states} at the base station, different transmission optimization schemes (e.g., power optimization, power-per-antenna optimization and max-min rate optimization) were investigated for designing linear precoders \cite{Farrokhi:JSAC98, Farrokhi:COM98, Bengtsson:Allerton99, Schubert:VT04, Wiesel:SP06, Yu:SP07}. On the other hand, recent works on  precoder design for {\em multi}-cell downlink transmissions  assumed that both data and channel state information of all users can be perfectly shared among base stations in real-time \cite{Zhang-Dai-2004,Foschini-2006,Zhang-ICC2007}. In  \cite{Foschini-2006,Zhang-Dai-2004}, coordinated base stations are simply regarded as a single large array with distributed antenna elements so that previously known single-cell precoding techniques  can be applied   to this scenario in a fairly straightforward manner. Instead, a more realistic model is considered in \cite{Zhang-ICC2007} which accounts for the fundamentally asynchronous nature of the interference due to the different propagation delays from the many base stations to each mobile. Practical concerns on the  complexity of the network infrastructure and synchronization requirements may permit coordination  only on a per-cluster basis as suggested in \cite{Andrews-2007,Gesbert-5}. Also,
the limited bandwidth of the backbone network connecting the base stations may prevent real-time data sharing; in this case,
each user can be served by only one base station, but the set of downlink precoders can still be optimized based on the inter-cell channel qualities \cite{Wei-Yu-2008,Lvent:Tw10}.

In practice, however, a base station can acquire only {\em imperfect} CSI which is contaminated with unknown errors. Motivated by this practical premise, the paradigm of {\em robust} optimization has been recently employed to address the problem of precoder designs for {\em single-cell} downlink transmissions \cite{Shenouda:JSAC08, Zhang:SP08, Vucic:SP09, Vucic:SP09_accepted}.
Imperfect CSI might be due erroneous channel estimation or quantization errors. For the former one, the uncertainty region of the CSI errors is modeled probabilistically, where it is assumed to be {\em unbounded} and distributed according to some known distribution. For the latter one, the uncertainty region of the CSI perturbations is assumed to be {\em bounded}. \cite{Shenouda:JSAC08} and \cite{Zhang:SP08} consider the probabilistic model and aim at optimizing a utility function by averaging over the entire uncertainty region. On the other hand, when  the CSI noise terms are bounded, a promising approach is   to design the precoders that yield worst-case guarantees, i.e., ensure worst-case robustness. Based on this notion, \cite{Vucic:SP09} examines the problems of mean-squared error ($\mse$) and  signal-to-interference-plus-noise ratio ($\sinr$) optimization for the multiple-input single-output (MISO) {\em single}-cell multiuser downlink transmission and \cite{Vucic:SP09_accepted} considers the problem of $\mse$ optimization for multiple-input multiple-output (MIMO) {\em single}-cell multiuser downlink systems.

In this paper we consider the more general model of {\em multi-cell} wireless networks, which hitherto has not been investigated for robust optimization, and treat the problem of {\em joint} robust transmission optimization for all cells. The significance of such multi-cell transmission optimization is that it incorporates the effects of {\em inter-cell} interference which is ignored when the cells optimize their transmissions independently. Furthermore, we incorporate a practical constraint which forbids real-time data sharing among base stations so that each user can  be served by only one base station. We adopt the bounded CSI noise model and design transceivers that optimize a {\em worst-case} quality of service measure. In particular, we focus on maximizing the   {\em worst-case weighted sum-rate} of the network and the  {\em minimum worst-case rate} of the network. For obtaining such transceiver designs, we offer centralized  and distributed algorithms with different levels of complexity and information exchange among the base stations. We also show that these problems can be translated into or approximated by convex problems that can be solved efficiently as semidefinite programs (SDP) with tractable computational complexity.

The remainder of the paper is organized as follows. In Section \ref{sec:model} we provide the system model. The formulations of the minimum worst-case rate and worst-case weighted sum-rate problems are discussed in Section \ref{sec:statement} and are treated in Sections \ref{sec:maxmin} and \ref{sec:sum}, respectively. We also offer distributed algorithms with limited cooperation and information exchange among the base stations in order to solve these problems. For the weighted sum-rate problem, we also provide some comments on the rate the the sum-rate scales with increasing $\snr$. Simulation results are given in Section \ref{sec:simulations} and Section \ref{sec:conclusions} concludes the paper.

\section{Transmission Model}
\label{sec:model}

We consider a multi-cell network with $M$ cells each with one base station (BS) that serves $K$ users. The BSs are equipped with $N$ transmit antennas and each user has one receive antenna. We denote B$_m$ as the BS of the $m^{th}$ cell and U$^k_m$ as the $k^{th}$ user in the $m^{th}$ cell for $m\in\{1,\dots,M\}$ and $k\in\{1,\dots,K\}$. We assume quasi-static flat-fading channels and denote the downlink channel from B$_n$ to U$^k_m$ by $\bh^k_{m,n}\in\mathbb{C}^{1\times N}$.

Let $\bx_m=[x_m^1,\dots,x_m^K]^T\in\mathbb{C}^{K\times 1}$ denote the information stream of B$_m$, intended for serving its designated users via spatial multiplexing and assume that $\bbe[\bx_m\bx_m^H]=\bI$. Prior to transmission by B$_m$, the information stream $\bx_m$ is {\em linearly} processed (precoded) by the precoding matrix $\bPhi_m\in\mathbb{C}^{N\times K}$. While non-linear precoding approaches can offer near optimal performance, they are not viable in practice. Alternatively, linear precoding approaches can achieve reasonable throughput performance (i.e., a small sub-optimality gap) with considerably lower complexity relative to non-linear precoding approaches and hence is the route adopted by emerging wireless standards such as 3GPP LTE and IEEE 802.16m.

We denote the $k^{th}$ column of $\bPhi_m$ by  $\bw^k_m\in\mathbb{C}^{N\times 1}$ which is the beam carrying the information stream intended for user U$^k_m$. By defining $f^k_m\in\mathbb{C}\;\backslash\;\{0\}$ as the single-tap receiver equalizer deployed by U$^k_m$, the received post-equalization signal at U$^k_m$ is given by
\begin{equation}\label{eq:model}
    y^k_m\;\dff\; \frac{1}{f^k_m}\Big(\sum_{n=1}^M\bh^k_{m,n}\bPhi_n\bx_n+z^k_m\Big),
\end{equation}
where $z^k_m\sim {\cal CN}(0,1)$ accounts for the additive white complex Gaussian noise. We assume that the users deploy single-user decoders for recovering their designated messages while suppressing the messages intended for other users as Gaussian interference. Therefore, the $\sinr$ of user U$^k_m$ (with the optimal equalizer) is given by
\begin{equation}\label{eq:sinr}
    \sinr^k_m\dff\frac{|\bh^k_{m,m}\bw^k_m|^2}{\sum_{l\neq k}|\bh^k_{m,m}\bw^l_m|^2 +\sum_{n\neq m}\sum_l|\bh^k_{m,n}\bw^l_n|^2+1}.
\end{equation}
Also we define $\tmse^k_m$ as the mean square-error ($\mse$) of user U$^k_m$ when it deploys the equalizer $f^k_m$, and it is given by
\begin{equation}
    \label{eq:tmse}
    \tmse^k_m\dff\bbe[|y^k_m-x^k_m|^2]=\frac{1}{|f^k_m|^2}\left(|\bh^k_{m,m}\bw^k_m-f^k_m|^2+\sum_{l\neq k}|\bh^k_{m,m}\bw^l_m|^2 +\sum_{n\neq m}\sum_l|\bh^k_{m,n}\bw^l_n|^2+1\right).
\end{equation}
We further define $\mse^k_m$ as the $\mse$ corresponding to the minimum mean-square error (MMSE) equalizer which minimizes the $\mse$ over all possible equalizers, i.e.,
\begin{equation}\label{eq:mse}
    \mse^k_m\dff\min_{f^k_m}\tmse^k_m.
\end{equation}
We assume that user U$^k_m$ perfectly knows its incoming channels, which includes all the channels $\{\bh^k_{m,n}\}$ for all choices of $n,k$. In contrast, each BS can acquire only {\em noisy} estimates of such channels corresponding to its designated receivers, i.e., B$_m$ knows the channels $\{\bh^k_{m,n}\}_{k,n}$ imperfectly. We denote the noisy estimate of the channel $\bh^k_{m,n}$ available at B$_m$ (and possibly other BSs via cooperation) by $\tilde \bh^k_{m,n}$ and define the channel estimation errors, which are {\em unknown} to the BSs, as
\begin{equation}\label{eq:error}
    \bDelta^k_{m,n}\dff \bh^k_{m,n} -\tilde \bh^k_{m,n}\ ,\quad\forall\;\; m,n\in\{1,\dots,M\}\;,\quad\mbox{and}\quad\forall\;\;k\in\{1,\dots,K\}.
\end{equation}
We assume that such channel estimation errors are bounded and confined within an origin-centered hyper-spherical region of radius $\epsilon^k_{m,n}$, i.e., $\|\bDelta^k_{m,n}\|_2\leq \epsilon^k_{m,n}$. There are two popular approaches for modeling the channel estimation errors: bounded and deterministic, and unbounded and random  with known distribution.
       The reasons for selecting the bounded model can be  summarized as follows.
       In the process of acquiring the channel state information (CSI) by the base-stations, there are two sources that induce uncertainty about the CSI. In this process the channels are first estimated by the mobiles and then are quantized and fed back to the base-stations. Therefore, there are two sources of CSI errors: estimation error and the quantization error. When the estimation is accurate enough but the amount of feedback bits available for feeding back the CSI  (which determine the size of the quantization codebook) is limited, the quantization error will be dominant error term. On the other hand, when there is no limit on the feedback rate and the channel estimates are not very accurate, then the error terms are dominated by the estimation errors
            In the emerging standards for the next-generation cellular wireless networks, the amount of bits reserved for quantized CSI feedback is small. Small number of feedback bits implies a coarse quantization codebook and hence deteriorates the quantization accuracy. Therefore, with a good estimator, there exists a high likelihood that the quantization error is dominant source of uncertainty about the CSI. Taking into account that the quantization errors are   bounded,   the model adopted here for the uncertainties about the CSI is justified.  It is also noteworthy that an issue with using the probabilistic model for the CSI perturbation in practical systems is identifying the right distribution since the Gaussian assumption on the estimation error need not be well justified. Finally, once a suitable distribution is identified for the probabilistic model, we can first find the probability that the CSI perturbations fall outside a bounded region (outage probability). Employing the outage probability in conjunction with the the analysis provided in our paper for the bounded CSI error model can also provide a direction for analyzing the systems with a probabilistic model for the CSI perturbation. We also note that all the results derived in the sequel can be readily extended to the case where the uncertainty regions are bounded hyper-ellipsoids.

In the sequel, for any matrix $\bA$ we use  $\|\bA\|_2$  to denote its Frobenius norm.

\section{Problem Statement}
\label{sec:statement}

We   consider  {\em multi}-cell downlink transmission with imperfect channel state information at the transmitter (imperfect CSIT).   The existing literature    concentrates mostly on the {\em single}-cell downlink system with imperfect CSIT.
          The significance of investigating multi-cell networks pertains to the fact that it allows us to design and optimize the network as an integrated entity. Optimizing each cell individually   ignores the impact of the cells on each other's performance, which in turn prevents from achieving a network-wise optimality.
          Also, analyzing multi-cell systems is not a straightforward generalization of the approaches known for the single-cell systems since there are two new challenges; there exists an additional source of interference, i.e., inter-cell interference.
              Moreover, some level of coordination/cooperation among the base-stations of different cells
            must be introduced.

We strive to optimize two network-wide performance measures through designing the precoding matrices $\{\bPhi_m\}$ and the receiver equalizers $\{f^k_m\}$. Such optimization heavily hinges on the {\em accuracy} of channel estimates available at the BSs. Due to the uncertainties about channels estimates, we adopt the notion of {\em robust} optimization in the {\em worst-case} sense \cite{Eldar:SP04,Boyd:book}. The solution of the worst-case robust optimization is feasible over the entire uncertainty region and provides the best guaranteed performance over all possible CSI errors.

Based on this notion of robustness, we treat the following rate optimization problems. One pertains to maximizing the  {\em worst-case weighted sum-rate} of the multi-cell network and the other one seeks to maximize the {\em minimum worst-case rate} in the network. Both optimizations are subject to individual power constraints for the BSs. Let us define $R^k_m$ as the rate assigned to user U$^k_m$  and denote the power budget for the BS B$_m$ by $P_m$. We also define $\bP\dff[P_1,\dots,P_M]$ as the vector of power budgets.

First we consider the  robust max-min rate problem which aims to maximize the  minimum worst-case rate of the network subject to the power budget $\bP$. Since the users are deploying single-user decoders, we have  $R^k_m=\log(1+\sinr^k_m)$. Therefore, this problem can be posed as
\begin{equation}\label{eq:rate}
    \CS(\bP)\dff\left\{
    \begin{array}{ll}
      \max_{\{\bPhi_m\}} & \min_{k,m}\;\min_{\{\bDelta^k_{m,n}\}}\;\sinr_m^k \\
      {\rm s.t.} & \|\bPhi_m\|^2_2 \leq P_m\quad\forall m
    \end{array}.
    \right.
\end{equation}
As the second problem, we consider optimizing the worst-case weighted sum-rate of the network. For a given set of positive weighting factors $\{\alpha^k_m\}$, where $\alpha^k_m$ is the weighting factor corresponding to the rate of user U$^k_m$, this problem is formalized as
\begin{equation}\label{eq:sum_rate}
    \R(\bP)\dff\left\{
    \begin{array}{ll}
      \max_{\{\bPhi_m\}} & \min_{\{\bDelta^k_{m,n}\}}\sum_{m=1}^M\sum_{k=1}^K\;\alpha^k_m\;R^k_m \\
      {\rm s.t.} & \|\bPhi_m\|^2_2\leq P_m\quad\forall  m
    \end{array}.
    \right.
\end{equation}

The weighted sum-rate utility function is most appropriate for the situations when we are interested in maximizing the network throughput while ensuring long term fairness. This utility does not include any hard minimum rate constraints for the  individual users. In an extreme case, for the benefit of the aggregate network throughput, some users might not be assigned any resource over some scheduling frames. Thus, such a utility function is useful when the users are delay tolerant (as in the case of best effort traffic) and would allow to be turned off or take small resources over some frames for the benefit of the network performance. By incorporating the weighting factors we can induce different priorities for the users. Moreover, these weights can be adapted over time to ensure long-term fairness.
         On the other hand, the notion of max-min rate optimization, which in each scheduling frame maximizes the rate of the weakest user in the network, guarantees a stricter (short-term) fairness among the users but at the cost of degraded network sum rate. Note that an optimal solution for the max-min rate utility will assign identical rates to all users.

In addition to solving the aforementioned optimization problems, we also analyze the degrees of freedom available in the multi-cell networks with imperfect CSIT.  The degrees of freedom metric has emerged as a popular tool for analyzing the sum-rates of the wireless networks in the asymptote of high $\snr$s and provides good insight into the sum-capacity of these systems for which the precise characterization of the capacity region is unknown. However,  existing works assume perfect CSI in determining the available degrees of freedom. Instead, in this work we  analyze the achievable degrees of freedom when only imperfect CSI is available.

We remark that the joint design of the optimal precoders may require that each BS acquires global CSI, which necessitates full cooperation (including full CSI exchange) among the BSs. In Sections \ref{sec:maxmin} and \ref{sec:sum} we provide centralized algorithms assuming that such full cooperation is feasible. In practice, however, full cooperation might not be implementable. In such cases we have to resort to distributed algorithms, that entail limited cooperation among the BSs. Towards this end, we also offer distributed algorithms in Sections \ref{sec:maxmin} and \ref{sec:sum} that involve limited information exchange and coordination among the base-stations, which in one instance comes at the cost of degraded performance compared to the corresponding centralized counterpart.

\section{Robust Max-Min Rate Optimization}
\label{sec:maxmin}

\subsection{Single-user Cells ($K=1$)}
\label{sec:maxmin_single}

In this subsection we assume that each BS is serving one user, i.e., $K=1$. Under this assumption, the downlink transmission model essentially becomes equivalent to a multiuser Gaussian interference channel with $M$ transmitters and $M$ respective receivers. For the ease of notation we omit the superscript $k$ in the subsequent analysis and discussions. When $K=1$ the precoder of BS B$_m$ consists of only one column vector which we refer to by $\bw_m$. For the given channel estimates $\{\tilde\bh_{m,n}\}$  we define $\tsinr_m\dff\min_{\{\bDelta_{m,n}\}}\sinr_m$ as the worst-case (smallest) $\sinr_m$ over the uncertainty regions.

A remark is warranted on the distinction between multi-cell single-user networks and single-cell multiuser networks. The fundamental difference is due to the different types of interference in these two systems. In single-cell multiuser systems there exists only intra-cell interference and is handled by a single transmitter (base-station). Base-station, knowing the downlink channels of all users, can jointly design the beamformers for all users. In the multi-cell single-user system, which can be also considered as a multi-user Gaussian interference channel, the users are facing only inter-cell interference. In these systems, as opposed to the single-cell systems, there are multiple transmitters (base-stations). Jointly designing the beamformers for all the users, therefore, requires some coordination and information exchange among the base-stations.

By introducing a slack variable $a>0$, the epigraph form of the robust max-min rate optimization problem $\CS(\bP)$ given in (\ref{eq:rate}) is given by
\begin{equation}\label{eq:rate2}
    \CS(\bP)=\left\{
    \begin{array}{ll}
      \max_{\{\bw_m\},a} & a \\
      {\rm s.t.} &\tsinr_m\geq a \quad \forall m, \\
       & \|\bw_m\|^2_2\leq P_m\quad\forall m.
    \end{array}
    \right.
\end{equation}
We proceed by finding the closed-form characterization of $\tsinr_m$. By recalling (\ref{eq:sinr}) we have
\begin{equation}\label{eq:sinr_tilde2}
    \tsinr_m  = \min_{\{\bDelta_{m,n}\}}\;\frac{|\bh_{m,m}\bw_m|^2}{\sum_{n\neq m}|\bh_{m,n}\bw_n|^2+1}=\frac{\min_{\bDelta_{m,m}}\;|\bh_{m,m}\bw_m|^2}{\sum_{n\neq m}\max_{\bDelta_{m,n}}|\bh_{m,n}\bw_n|^2+1},
\end{equation}
where the second equality holds by noting that the channels $\bh_{m,i}$ and $\bh_{m,j}$ for $i\neq j$ have independent uncertainties and finding the worst-case uncertainties can be decoupled. Therefore, finding the worst-case $\sinr_m$ can be decoupled into finding the worst-case (smallest) numerator term and the worst-case (largest) denominator terms. In order to further simplify $\tsinr_m$ we use the result of the following lemma. The proof is straightforward and is a simple extension of a result in \cite{Gershman:TSP03} which considers  robust beamforming for a point-to-point link with colored interference at the receiver.

\begin{lemma}
\label{lemma:robust_sinr}
For any given $\bh\in\mathbb{C}^{1\times N}$,  $\bw\in\mathbb{C}^{N\times 1}$, $\epsilon\in\mathbb{R}^{+}$, and  positive definite matrix $\bQ$; $g_{\min}$ and $g_{\max}$ defined as
\begin{equation*}
    g_{\min}\dff\left\{
    \begin{array}{ll}
      \min_{\bx} & |\bh\bw+\bx\bw|^2\\
      {\rm s.t.} & \sqrt{\bx\bQ\bx^H} \leq\epsilon
    \end{array}
    \right. ,
    \quad\mbox{and}\quad
    g_{\max}\dff\left\{
    \begin{array}{ll}
      \max_{\bx} & |\bh\bw+\bx\bw|^2\\
      {\rm s.t.} & \sqrt{\bx\bQ\bx^H} \leq\epsilon
    \end{array}
    \right. ,
\end{equation*}
are given by
\begin{equation*}
    g_{\min}=\Big|\Big(|\bh\bw|-\epsilon\sqrt{\bw^H\bQ^{-1}\bw}\;\Big)^+\Big|^2 , \quad\mbox{and}\quad g_{\max}=\Big||\bh\bw|+\epsilon\sqrt{\bw^H\bQ^{-1}\bw}\Big|^2,
\end{equation*}
where $(x)^+=\max\{0,x\},\;\forall\;x\in\mathbb{R}$.
\end{lemma}
By recalling \eqref{eq:sinr_tilde2} and invoking the result of the lemma above for the choice of $\bQ=\bI$, $\tsinr_m$ can be further simplified as
\begin{equation}\label{eq:sinr_tilde}
    \tsinr_m=\frac{\big|(|\tilde\bh_{m,m}\bw_m|-\epsilon_{m,m}\|\bw_m\|_2)^+\big|^2} {\sum_{n\neq m} \big||\tilde\bh_{m,n}\bw_n|+\epsilon_{m,n}\|\bw_n\|_2\big|^2+1}.
\end{equation}
We are interested in the scenarios where $\forall\;m,\;\|\tilde\bh_{m,m}\|>\epsilon_{m,m},\;$ so that $\CS(\bP)>0$. Given the closed-form characterization of $\tsinr_m$,   solving $\CS(\bP)$ can be facilitated by solving a power optimization problem defined as
\begin{equation}\label{eq:f}
    \CP(\bP,a)\dff\left\{
    \begin{array}{ll}
      \min_{\{\bw_m\},b} & b \\
      {\rm s.t.} &\tsinr_m\geq a \quad\forall m, \\
       & \frac{\|\bw_m\|_2}{\sqrt{P_m}}\leq b\quad\forall m\ .
    \end{array}
    \right.
\end{equation}
In this problem the maximum weighted power consumed by the base-stations is minimized subject to a quality of service guarantee for all users. The constraint $\tsinr_m\geq a$ assures that all users receive a worst-case $\sinr$ which is not smaller than $a$ and the constraint $\frac{\|\bw_m\|_2}{\sqrt{P_m}}\leq b$ guarantees that the maximum weighted power consumption is minimized. This can be considered as the dual of the problem which strives to maximize the quality of service for the users for a given power budget. The connection between $\CS(\bP)$ and $\CP(\bP,a)$ is established in the following useful theorem.
\begin{theorem}
\label{th:connection}
For any given power budget $\bP$, $\CP(\bP,a)$ is strictly increasing and continuous in $a$ at any strictly feasible $a$ and is related to $\CS(\bP)$ via
\begin{equation*}
    \CP(\bP, \CS(\bP))=1.
\end{equation*}
\end{theorem}
\begin{proof}
See Appendix \ref{app:th:connection}.
\end{proof}
Strict monotonicity and continuity of $\CP(\bP,a)$ in $a$ at any strictly feasible $a$ provides that there exists a unique $a^*$ satisfying $\CP(\bP,a^*)=1$. Hence, taking into account Theorem \ref{th:connection} establishes that solving $\CS(\bP)$ boils down to finding $a^*$ that satisfies $\CP(\bP,a^*)=1$. Due to monotonicity and continuity of $\CP(\bP,a)$, finding $a^*$ can be implemented via a simple iterative bi-section search. Each iteration  requires solving $\CP(\bP,a)$ for a different value  of $a$.  We demonstrate that $\CP(\bP,a)$ can be cast as a convex problem with a computationally efficient solution.
\begin{theorem}\label{th:SU1}
Problem $\CP(\bP,a)$ can be posed as an SDP.
\end{theorem}
\begin{proof}
See Appendix \ref{app:th:SU1}.
\end{proof}
Such a procedure for solving $\CS(\bP)$ necessitates the BSs to be fully cooperative such that each BS can acquire estimates of all network-wide channel states.

\subsection{Multi-user Cells ($K>1$)}
\label{sec:maxmin_multi}

In this subsection we consider downlink transmissions serving more than one user in each cell ($K>1$). The major difference between the analysis for multiuser cells and that of single-user cells arises from  the different characterizations of their corresponding worst-case $\sinr$s. By defining $\tsinr^k_m$ as the worst-case $\sinr^k_m$ in (\ref{eq:sinr}), we have
\begin{equation}\label{eq:sinr_tilde_MU}
    \tsinr^k_m=\min_{\{\bDelta^k_{m,n}\}}\frac{|\bh^k_{m,m}\bw^k_m|^2}{\sum_{l\neq k}|\bh^k_{m,m}\bw^l_m|^2 +\sum_{n\neq m}\sum_l|\bh^k_{m,n}\bw^l_n|^2+1}.
\end{equation}
Unlike the single-user setup, when $K>1$ the numerator and the summands of the denominator of $\tsinr^k_m$ have a common uncertainty term. Therefore, finding $\tsinr^k_m$ cannot be decoupled into finding the worst-case numerator and the worst-case terms in the denominator independently. To the best of our knowledge, handling constraints on  such worst-case $\sinr$s even in single-cell downlink transmissions is not mathematically tractable \cite{Vucic:SP09} and the robust design of linear precoders for these systems is carried out {\em suboptimally} \cite{Vucic:SP09, Vucic:SP09_accepted}. In the sequel we also propose suboptimal approaches for solving the robust max-min rate optimization in multi-cell networks.

We offer two suboptimal approaches for solving $\CS(\bP)$.  In the first approach we find a {\em lower} bound on the worst-case $\sinr$ and in the formulation of $\CS(\bP)$ replace each worst-case $\sinr$ with its corresponding lower bound. Similar to the single-user cells setup in Section \ref{sec:maxmin_single}, this approximate problem can be solved efficiently through solving a counterpart power optimization problem. In the second approach, we convert the robust max-min rate optimization problem into a robust min-max $\mse$ optimization problem and find an upper bound on the maximum worst-case $\mse$ which in turn provides a lower bound on the minimum worst-case rate.

\subsubsection{Solving via Power Optimization}
\label{sec:maxmin_power}

The worst-case value of $\sinr^k_m$, which we denote by $\tsinr^k_m$, is not mathematically tractable. Consequently, we find lower bounds on the worst-case $\sinr$s as follows. We define
\begin{equation*}
    \bsinr^k_m\dff\frac{\min_{\bDelta^k_{m,m}}|\bh^k_{m,m}\bw^k_m|^2}{\max_{\bDelta^k_{m,m}}\sum_{l\neq k}|\bh^k_{m,m}\bw^l_m|^2 +\sum_{n\neq m}\max_{\bDelta^k_{m,n}}\sum_l|\bh^k_{m,n}\bw^l_n|^2+1},
\end{equation*}
for which we clearly have  $\bsinr^k_m\leq \tsinr^k_m$. By applying Lemma \ref{lemma:robust_sinr} we find that
\begin{equation}\label{eq:sinr_bar}
    \bsinr^k_m=\frac{\big|(|\tilde\bh_{m,m}^k\bw_m^k|-\epsilon_{m,m}^k\|\bw_m^k\|_2)^+\big|^2} {\max_{\bDelta^k_{m,m}} \bh^k_{m,m}\bPsi_{m,k}(\bPsi_{m,k})^H(\bh^k_{m,m})^H +
    \sum_{n\neq m}\max_{\bDelta^k_{m,n}} \bh^k_{m,n}\bPhi_{n}(\bPhi_{n})^H(\bh^k_{m,n})^H+1},
\end{equation}
where we have defined $\bPsi_{m,k}\dff[\bw^1_m\dots,\bw^{k-1}_m,\bw^{k+1}_m,\dots,\bw^K_m]$.

By introducing the slack variable $a$ and invoking the lower bounds on the $\sinr$s given in (\ref{eq:sinr_bar}), a lower bound on the robust max-min rate is obtained as follows.
\begin{equation}\label{eq:rate3}
    \CS_1(\bP)\dff\left\{
    \begin{array}{ll}
      \max_{\{\bPhi_m\},a} & a \\
      {\rm s.t.} &\bsinr^k_m\geq a \quad \forall m,k, \\
       & \|\bPhi_m\|^2_2\leq P_m\quad\forall m.
    \end{array}
    \right.
\end{equation}
Similar to the single-user scenario, solving $\CS_1(\bP)$ can be carried out by alternatively solving a power optimization problem in conjunction with a linear bi-section search. The power optimization of interest with per BS power constraints is given by
\begin{equation}\label{eq:f2}
    \CP_1(\bP,a)\dff\left\{
    \begin{array}{ll}
      \min_{\{\bPhi_m\},b} & b \\
      {\rm s.t.} &\bsinr^k_m\geq a \quad\forall m,k, \\
       & \frac{\|\bPhi_m\|_2}{\sqrt{P_m}}\leq b\quad\forall m.
    \end{array}
    \right.
\end{equation}
The result in Theorem \ref{th:connection} can be extended for multiuser cell setup ($K>1$) in order to establish the connection between $\CS_1(\bP)$ and $\CP_1(\bP,a)$. The proof is  omitted for brevity.
\begin{theorem}
\label{th:connection2}
For any given power budget $\bP$, $\CP_1(\bP,a)$ is strictly increasing and continuous in $a$ at any strictly feasible $a$  and is related to $\CS_1(\bP)$ via
\begin{equation}\label{eq:connection}
    \CP_1(\bP, \CS_1(\bP))=1.
\end{equation}
\end{theorem}
\begin{figure*}[t]
{\small \begin{minipage}[t!]{6.5 in}
\rule{\linewidth}{0.3mm}\vspace{-.05in}
{\bf Algorithm 1 } - Robust Max-Min $\sinr$ Optimization via Power Optimization ($K\geq 1$) \vspace{-.15in}\\
\rule{\linewidth}{0.2mm}
{ {\rm
\begin{tabular}{ll}
   \;1:& Input $\bP$ and $\{\tilde\bh^{k}_{m,n},\epsilon_{m,n}^k\}$ \\
   \;2:& Initialize $a_{\min}=0$ and $a_{\max}=\min_{m,k}\{P_m|(\|\tilde\bh^k_{m,m}\|_2-\epsilon^k_{m,m})^+|^2\}$\\
   \;3:& $a_0\leftarrow a_{\min}$\\
   \;4:& {\bf repeat}\\
   \;5:& \quad Solve $\CP_1(\bP,a_0)$ and obtain $\{\bPhi_m\}$\\
   \;6:& \quad \textbf{\bf if} $\CP_1(\bP,a_0)\leq 1$\\
   \;7:& \quad\quad $a_{\min}\leftarrow a_0$\\
   \;8:& \quad \textbf{\bf else}\\
   \;9:& \quad \quad $a_{\max}\leftarrow a_0$\\
   10:& \quad \textbf{\bf end if}\\
   \;11:& \quad $a_0\leftarrow(a_{\min}+a_{\max})/2$\\
   12:&{\bf until} $a_{\max}-a_{\min}\leq \delta$\\
   13:& Output $\CS_1(\bP)=a_0$ and $\{\bPhi_m\}$
\end{tabular}}}\\
\rule{\linewidth}{0.3mm}
\end{minipage}}
\end{figure*}

We construct Algorithm 1 which solves $\CS_1(\bP)$ by solving $\CP_1(\bP,a)$ combined with a bi-section line search. The optimality of Algorithm 1  and its convergence follows from the monotonicity and continuity of $\CP_1(\bP,a)$ at any feasible $a$. Similar to Theorem \ref{th:SU1} we demonstrate that $\CP_1(\bP,a)$ has a computationally efficient solution.
\begin{theorem}
\label{th:sdp_m}
Problem $\CP_1(\bP,a)$ can be posed as an SDP.
\end{theorem}
\begin{proof}
See Appendix \ref{app:th:sdp_m}.
\end{proof}
Note that when $K=1$, $\CS_1(\bP)$ is identical to $\CS(\bP)$ and an optimal solution to the latter problem is obtained using Algorithm 1.

\subsubsection{Solving via $\mse$ Optimization}
\label{sec:maxmin_mse}

First we transform the robust max-min rate optimization problem into a robust min-max $\mse$ optimization problem by using the fact that
\begin{equation*}
    \mse^k_m=\frac{1}{1+\sinr^k_m},
\end{equation*}
where we recall that $\mse^k_m$ is the $\mse$ of   user U$^k_m$ when it deploys the MMSE equalizer.
Consequently,   the worst-case $\mse$ corresponding to the worst-case $\sinr$ is given as
\begin{equation}
    \label{eq:mse_wc}
    \max_{\{\bDelta^k_{m,n}\}}\mse^k_m=\frac{1}{1+\min_{\{\bDelta^k_{m,n}\}}\sinr^k_m}=\frac{1}{1+\tsinr^k_m}.
\end{equation}
By recalling the problem $\CS(\bP)$ given in (\ref{eq:rate}) and taking into account the representation of $\mse^k_m$ given in (\ref{eq:mse}) and the worst-case $\mse$ given in (\ref{eq:mse_wc}), the robust max-min rate optimization problem can be solved by equivalently solving
\begin{equation}\label{eq:mse_opt}
    \CS_2(\bP)\dff\left\{
    \begin{array}{ll}
      \min_{\{\bPhi_m\}} & \max_{k,m}\max_{\{\bDelta^k_{m,n}\}}\min_{f^k_m}\tmse^k_m \\
      {\rm s.t.} & \|\bPhi_m\|^2_2\leq P_m\quad\forall  m.
    \end{array}
    \right.
\end{equation}
We proceed by finding an upper bound on $\CS_2(\bP)$ which in turn results in a lower bound on $\CS(\bP)$. By invoking the inequality\footnote{Note that for any function $f(x,y)$ we have $\min_x\max_yf(x,y)\geq \max_y\min_xf(x,y)$. To see this, define $(x_0,y_0)=\arg\min_x\max_y f(x,y)$ and $(x_1,y_1)= \arg\max_y\min_xf(x,y)$. Therefore, $f(x_0,y_0)=\max_yf(x_0,y)\geq f(x_0,y_1)\geq \min_xf(x,y_1)=f(x_1,y_1)$.}
\begin{equation}
    \label{eq:inequality}
    \max_{\{\bDelta^k_{m,n}\}} \min_{f^k_m}\tmse^k_m\leq  \min_{f^k_m}\max_{\{\bDelta^k_{m,n}\}}\tmse^k_m,
\end{equation}
and introducing the slack variable $a\in\mathbb{R}^{+}$, we find the following upper bound on $\CS_2(\bP)$.
\begin{equation}\label{eq:mse_opt2}
    \bar\CS_2(\bP)\dff\left\{
    \begin{array}{ll}
      \min_{\{\bPhi_m,f^k_m\},a} & a\\
      {\rm s.t.} & \max_{\{\bDelta^k_{m,n}\}}\tmse^k_m\leq a^2 \quad \forall \; k,m,\\
      & \|\bPhi_m\|^2_2\leq P_m\quad\forall  m.
    \end{array}
    \right.
\end{equation}
This problem for {\em single}-cell MIMO downlink transmissions is studied in \cite{Vucic:SP09_accepted} where it is solved suboptimally via an iterative algorithm based on the alternating optimization (AO) principle.
Here we show that the problem in (\ref{eq:mse_opt2}) is equivalent to a generalized eigenvalue problem (GEVP) (c.f. \cite{Wiesel:SP06}) which can be solved efficiently.

\begin{theorem}
\label{th:rate_m_mse}
The problem $\bar\CS_2(\bP)$ can be optimized  efficiently as a GEVP.
\end{theorem}
\begin{proof}
See Appendix \ref{app:th:GEVP}.
\end{proof}

\subsection{Limited Cooperation}
\label{sec:maxmin_limited}

In this subsection we propose  {\em distributed} algorithms for the networks not supporting full CSI exchange between the BSs. The first distributed algorithm we propose (Algorithm 2) requires {\em limited} information exchange between the BSs and each BS designs its precoders independently of others. The cost incurred for enabling such  limited information exchange and distributed processing is the degraded performance  compared with the centralized algorithm.

The underlying notion of this distributed algorithm is to successively update the precoder of one BS at-a-time while keeping rest unchanged. More specifically, at the $m^{th}$ iteration,   all precoders $\{\bPhi_n\}_{n\neq m}$ are fixed and only BS B$_m$ updates its precoder by maximizing the worst-case smallest rate of the $m^{th}$ cell. By recalling \eqref{eq:rate3}, the optimization problem solved by B$_m$ is given by
\begin{equation}\label{eq:rate_dist}
    \CS_{1,m}(\bP)\dff\left\{
    \begin{array}{ll}
      \max_{\bPhi_m,a} & a \\
      {\rm s.t.} &\bsinr^k_m\geq a \quad \forall k \\
      & \|\bPhi_m\|^2_2\leq P_m,\\
      & \bPhi_n \mbox{ are fixed for } n\neq m.
    \end{array}
    \right.
\end{equation}
Similar to the approach of Section \ref{sec:maxmin_power} it can be readily verified that $\CS_{1,m}(\bP)$ can be solved through the following power optimization problem,
\begin{equation}\label{eq:f3}
    \CP_{1,m}(\bP,a)\dff\left\{
    \begin{array}{ll}
      \min_{\bPhi_m,b} & b \\
      {\rm s.t.} & \bsinr^k_m\geq a \quad \forall k,\\
       & \frac{\|\bPhi_m\|_2}{\sqrt{P_m}}\leq b\quad,\\
       & \bPhi_n \mbox{ are fixed for } n\neq m.
    \end{array}
    \right.
\end{equation}
which is connected to the original problem $\CS_{1,m}(\bP)$ as follows.
\begin{coro}
\label{cor:connection}
For any given power budget $\bP$, $\CP_{1,m}(\bP,a)$ is strictly increasing and continuous in $a$ at any strictly feasible $a$ and is related to $\CS_{1,m}(\bP)$ via
\begin{equation}\label{eq:connection3}
    \CP_{1,m}(\bP, \CS_{1,m}(\bP))=1\quad\mbox{for}\quad m=1,\dots,M.
\end{equation}
\end{coro}
To compute $\bsinr^k_m$ in (\ref{eq:sinr_bar}) it is seen that $B_m$  needs to know $\bW_n=\bPhi_n\bPhi_n^H,\;n\neq m$. Using $\{\bW_n\}$,  B$_m$ can solve $\CS_{1,m}(\bP)$ optimally and obtains $\bPhi^*_m$ through solving $\CP_{1,m}(\bP,a)$ in conjunction with a linear bi-section search (Algorithm 2, line 8). Note that solving $\CS_{1,m}(\bP)$ optimizes the minimum worst-case rate locally in the $m^{th}$ cell and does not necessarily leads to a boost in the network utility function. As a result, B$_m$ is allowed to update its precoder to $\bPhi^*_m$ only if such update results in a network-wide improvement (lines 9-13).

The successive updates of the precoders continue until no precoder can be further updated unilaterally. The convergence to such point is guaranteed by noting that the algorithm imposes the constraint that B$_m$ can update its precoder only if it results in network-wide improvement.  We also note that another variation of Algorithm 2 is also possible. In this variation, instead of fixing $\{1,\cdots,M\}$ as the order of processing (i.e., the order in which the BSs attempt to update their precoders) as done in
Algorithm 2, we can employ a greedy approach. In particular, at each iteration each BS can compute its precoder (assuming precoders of other BSs to be fixed). Then in a bidding phase, each BS can  broadcast its choice and only the choice which maximizes the network minimum worst-case rate is accepted by all BSs.

We next propose another distributed algorithm that can {\em optimally} solve the optimization problem in (\ref{eq:rate3}) by introducing more auxiliary variables and using dual decomposition. This algorithm involves a higher level of inter BS signaling than the previous distributed procedure. Such an approach of  introducing more auxiliary variables and using dual decomposition has been employed for the multi-cell uplink with perfect CSI  in \cite{KYang:SP09} and more recently in \cite{Tolli-Globe2007} over the multi-cell downlink also with perfect CSI.
First, we re-write (\ref{eq:rate3}) as
\begin{equation}\label{eq:rate3N}
   \left\{
    \begin{array}{ll}
      \max_{\{\bPhi_m,\beta^k_{m,n}\},a} & a \\
     & {\rm s.t.} \frac{\big|(|\tilde\bh_{m,m}^k\bw_m^k|-\epsilon_{m,m}^k\|\bw_m^k\|_2)^+\big|^2} {\max_{\bDelta^k_{m,m}} \bh^k_{m,m}\bPsi_{m,k}(\bPsi_{m,k})^H(\bh^k_{m,m})^H +\sum_{n\neq m}(\beta^k_{m,n})^2
    +1} \geq a \quad \forall m,k, \\
       & \|\bPhi_m\|^2_2\leq P_m\quad\forall m \\
       & \max_{\bDelta^k_{m,n}} \bh^k_{m,n}\bPhi_{n}(\bPhi_{n})^H(\bh^k_{m,n})^H\leq(\beta^k_{m,n})^2\quad \forall k,m\neq n.
    \end{array}
    \right.
\end{equation}
To solve (\ref{eq:rate3N}) we propose a bi-section search over $a$ in which
for any fixed $a$, we solve the following problem
\begin{equation}\label{eq:rate3Nb}
   \left\{
    \begin{array}{ll}
      \min_{\{\bPhi_m,\beta^k_{m,n}\}} & \sum_m\|\bPhi_m\|_2^2 \\
     & {\rm s.t.} \frac{\big|(|\tilde\bh_{m,m}^k\bw_m^k|-\epsilon_{m,m}^k\|\bw_m^k\|_2)^+\big|^2} {\max_{\bDelta^k_{m,m}} \bh^k_{m,m}\bPsi_{m,k}(\bPsi_{m,k})^H(\bh^k_{m,m})^H +\sum_{n\neq m}(\beta^k_{m,n})^2
    +1} \geq a \quad \forall m,k, \\
       & \|\bPhi_m\|^2_2\leq P_m\quad\forall m \\
       & \max_{\bDelta^k_{m,n}} \bh^k_{m,n}\bPhi_{n}(\bPhi_{n})^H(\bh^k_{m,n})^H\leq(\beta^k_{m,n})^2\quad \forall m\neq n,k.
    \end{array}
    \right.
\end{equation}
Next, for each BS $m$, we define  variables $\beta^{k,m}_{m,n},\beta^{j,m}_{n,m}$ which denote its copies of $\beta^k_{m,n},\beta^{j}_{n,m}$, respectively. Also, let $\bbeta^{(m)}$ be the vector formed by collecting all such variables.
 Then, we can re-write (\ref{eq:rate3Nb}) as
\begin{equation}\label{eq:rate3Nc}
    \left\{
    \begin{array}{ll}
      \min_{\{\bPhi_m,\bbeta^{(m)}\}} & \sum_m\|\bPhi_m\|_2^2 \\
     & {\rm s.t.} \frac{\big|(|\tilde\bh_{m,m}^k\bw_m^k|-\epsilon_{m,m}^k\|\bw_m^k\|_2)^+\big|^2} {\max_{\bDelta^k_{m,m}} \bh^k_{m,m}\bPsi_{m,k}(\bPsi_{m,k})^H(\bh^k_{m,m})^H +\sum_{n\neq m}(\beta^{k,m}_{m,n})^2
    +1} \geq a \quad \forall  k,m    \\
       &\max_{\bDelta^j_{n,m}}\bh^j_{n,m}\bPhi_{m}(\bPhi_{m})^H(\bh^j_{n,m})^H\leq (\beta^{j,m}_{n,m})^2\quad \forall j,n\neq m\\
       & \|\bPhi_m\|^2_2\leq P_m  \quad \forall m \\
       & \beta^{k,m}_{m,n}=\beta^{k,n}_{m,n}   \quad \forall k,m\neq n.
    \end{array}
    \right.
\end{equation}
Similar to the proof of Theorem \ref{th:sdp_m}, it can be shown that (\ref{eq:rate3Nc}) is equivalent to a (convex) SDP. Therefore, (\ref{eq:rate3Nc}) being convex provides that the optimal primal value of (\ref{eq:rate3Nc}) and that of its Lagrangian dual are equal. In other words, the duality gap is zero and strong duality holds for (\ref{eq:rate3Nc}) provided Slater's condition is also satisfied. Define dual variables $\{\lambda^{k}_{m,n}\}$ and let $\blambda$ denote the vector formed by collecting all such variables. Consider the following partial Lagrangian
\begin{equation}
     L(\{\bPhi_m,\bbeta^{(m)}\},\blambda) \dff    \sum_m\|\bPhi_m\|_2^2 +\sum_{m\neq n}\sum_{k} \lambda^k_{m,n}(\beta^{k,m}_{m,n}-\beta^{k,n}_{m,n})
\end{equation}
 and the dual function
\begin{equation}\label{eq:rate3Nd}
   g(\blambda) \dff \left\{
    \begin{array}{ll}
         \min_{\{\bPhi_m,\bbeta^{(m)}\}} & \sum_m\|\bPhi_m\|_2^2  +\sum_{m\neq n}\sum_{k} \lambda^k_{m,n}(\beta^{k,m}_{m,n}-\beta^{k,n}_{m,n})\\
     & {\rm s.t.} \frac{\big|(|\tilde\bh_{m,m}^k\bw_m^k|-\epsilon_{m,m}^k\|\bw_m^k\|_2)^+\big|^2} {\max_{\bDelta^k_{m,m}} \bh^k_{m,m}\bPsi_{m,k}(\bPsi_{m,k})^H(\bh^k_{m,m})^H +\sum_{n\neq m}(\beta^{k,m}_{m,n})^2
    +1} \geq a \quad \forall  k,m    \\
       &\max_{\bDelta^j_{n,m}}\bh^j_{n,m}\bPhi_{m}(\bPhi_{m})^H(\bh^j_{n,m})^H\leq(\beta^{j,m}_{n,m})^2\quad \forall j,n\neq m\\
       & \|\bPhi_m\|^2_2\leq P_m  \quad \forall m \\
    \end{array}
    \right.
\end{equation}
 Notice  that the dual problem splits into $M$ smaller problems of the form
\begin{equation}\label{eq:rate3Ne}
    \left\{
    \begin{array}{ll}
         \min_{\bPhi_m,\{\bbeta^{(m)}\}} &  \|\bPhi_m\|^2_2  +\sum_{n:n\neq m}(\sum_{k} \lambda^k_{m,n}\beta^{k,m}_{m,n}-\sum_j\lambda^j_{n,m}\beta^{j,m}_{n,m})\\
     & {\rm s.t.} \frac{\big|(|\tilde\bh_{m,m}^k\bw_m^k|-\epsilon_{m,m}^k\|\bw_m^k\|_2)^+\big|^2} {\max_{\bDelta^k_{m,m}} \bh^k_{m,m}\bPsi_{m,k}(\bPsi_{m,k})^H(\bh^k_{m,m})^H +\sum_{n\neq m}(\beta^{k,m}_{m,n})^2
    +1} \geq a \quad \forall  k  \\
       &\max_{\bDelta^j_{n,m}}\bh^j_{n,m}\bPhi_{m}(\bPhi_{m})^H(\bh^j_{n,m})^H\leq(\beta^{j,m}_{n,m})^2\quad \forall j,n:n\neq m\\
       & \|\bPhi_m\|^2_2\leq P_m.\\
    \end{array}
    \right.
\end{equation}
Using the arguments provided in the preceding sections, each of the smaller problem in (\ref{eq:rate3Ne}) can be shown to be equivalent to an SDP.
Invoking the strong duality, we can recover the primal optimal solution by solving the dual problem
 $\max_{\blambda}\{g(\blambda)\}$.  The latter problem can be also solved in a distributed manner via the sub-gradient method. In particular, suppose $\{\hat{\bbeta}^{(m)}\}$ are the optimized variables obtained upon solving the decoupled optimization problems in (\ref{eq:rate3Ne}). Then the dual variables can be updated using a sub-gradient as
 $\lambda^{k}_{m,n}\to \lambda^{k}_{m,n} +\mu(\hat{\beta}^{k,m}_{m,n}-\hat{\beta}^{k,n}_{m,n}),\;\forall\;k,m\neq n$, where
  $\mu$ is a positive step size parameter. Note that updating $\lambda^{k}_{m,n}$ involves exchanging $\hat{\beta}^{k,m}_{m,n},\hat{\beta}^{k,n}_{m,n}$ between BSs $m,n$, respectively.
  Finally, we note that a speed-up can be obtained at the cost of some sub-optimality by forcing equality after a few steps of the sub-gradient method. In particular, we can set
  both $\beta^{k,m}_{m,n},\beta^{k,n}_{m,n}$ to be equal to $\frac{\hat{\beta}^{k,m}_{m,n}+\hat{\beta}^{k,n}_{m,n}}{2}$ for all $k,m\neq n$ and then concurrently optimize the $M$ decoupled problems in (\ref{eq:rate3Nc}) over $\{\bPhi_m\}$. The current choice of $a$ is declared feasible if and only if all of the problems are feasible.

\begin{figure*}[t]
{\small \begin{minipage}[t!]{6.5 in}
\rule{\linewidth}{0.3mm}\vspace{-.05in}
{\bf Algorithm 2} - Distributed Robust Max-Min $\sinr$ Optimization \vspace{-.15in}\\
\rule{\linewidth}{0.2mm}
{ {\rm
\begin{tabular}{ll}
   \;1:& \textbf{for} $m=1,\dots,M$ \textbf{do}\\
   \;2:& \quad Input $P_m$\\
   \;3:& \quad $B_m$ initializes $\bPhi_m=\frac{P_m}{K}\left[\frac{(\tilde\bh^1_{m,m})^H}{\|\tilde\bh^1_{m,m}\|_2}, \dots, \frac{(\tilde\bh^K_{m,m})^H}{\|\tilde\bh^K_{m,m}\|_2}\right]$ and broadcasts $\bW_m=\bPhi_m\bPhi_m^H$\\
   \;4:& \textbf{end for}\\
   \;5:& \quad Using $\bPhi_m,\{\bW_n\}_{n\neq m}$, each $B_m$ computes $\bsinr^{l}_m,\;\forall\;l$\\
   \;6:& \textbf{repeat}\\
   \;7:& \quad \textbf {for} $m=1,\dots,M$ \textbf{do}\\
   \;8:& \quad\quad B$_m$ solves $\CS_{1,m}(\bP)$   and obtains $\bPhi^*_m$; \\
   \;9:& \quad\quad  B$_m$ broadcasts $\bW_m^*=\bPhi^*_m(\bPhi^*_m)^H$ \\
   \;10:& \quad\quad Each B$_n$, $n\neq m$, computes $\bsinr^{l*}_n,\;\forall\;l$ based on $\bW^*_m,\bPhi_n$ and $\{\bW_j\}_{j\neq m,n}$\\
   11:& \quad \quad \textbf{if} $\min_{l}\{\bsinr^{l*}_n\}<\min_{l}\{\bsinr^l_n\}$ then B$_n$ sends an error message to B$_m$\\
   12:& \quad \quad \textbf{if}  B$_m$ receives no error message then it sets $\bPhi_m \leftarrow \bPhi^*_m$ and  broadcasts an update message \\
   13:& \quad \quad Upon receiving the update message each B$_n$, $n\neq m$,  sets $\bW_m \leftarrow \bW^*_m$ and updates  $\bsinr^{l}_n,\;\forall\;l$ \\
   14:& \quad \textbf{end for}\\
   15:&{\bf until} no further precoder update is possible\\
   16:& Output $\{\bPhi_m\}$
\end{tabular}}}\\
\rule{\linewidth}{0.3mm}
\end{minipage}}
\end{figure*}

\section{Robust Weighted Sum-rate Optimization}
\label{sec:sum}
\subsection{Full Cooperation}
By recalling the definitions in (\ref{eq:error}) and taking into account that the uncertainty regions corresponding to $\sinr^k_m$ and $\sinr^l_n$ for $m\neq n$ or $l\neq k$ are disjoint, finding the worst-case $\sinr$ for each user can be carried out independently of the rest. Hence, by recalling that the worst-case $\sinr$ of user U$^k_m$ is denoted by $\tsinr^k_m$, the problem $\R(\bP)$ is given by
\begin{equation}\label{eq:sum_rate2}
    \R(\bP)=\left\{
    \begin{array}{ll}
      \max_{\{\bPhi_m\}} & \sum_{m=1}^M\sum_{k=1}^K\;\alpha^k_m\;\log(1+\tsinr^k_m) \\
      {\rm s.t.} & \|\bPhi_m\|^2_2\leq P_m\quad\forall  m.
    \end{array}
    \right.
\end{equation}
The problem $\R(\bP)$ as posed above, is not a convex problem. Optimal precoder design based on maximizing the weighted sum-rate even when the BSs have {\em perfect} CSI is known to be an NP hard problem. Thus, even in this case only efficient techniques yielding locally optimal solutions \cite{Lvent:Tw10} can be obtained. Clearly, the robust weighted-sum rate problem is also NP hard and hence good sub-optimal algorithms are of particular interest. Here we leverage an approach developed in \cite{Agarwal:ISITA08} (which considers the single-cell scenario and designs a provably convergent algorithm yielding locally optimal solutions) and propose a suboptimal solution by obtaining a {\em conservative} approximation of the problem $\R(\bP)$. This approximation provides a lower bound on $\R(\bP)$.

To start we define the set of functions  $\{S^k_m(u):\mathbb{R}\rightarrow\mathbb{R}\}$ as
\begin{equation*}
    S^k_m(u)\dff\alpha^k_m\;u-\frac{\alpha^k_m}{1+\tsinr^k_m}\exp(u-1),\quad\mbox{for}\quad m=1, \dots, M,\quad k=1,\dots,K,\quad\mbox{and}\quad u\in\mathbb{R}.
\end{equation*}
It can be readily verified that we have
\begin{equation}
    \label{eq:u}
    \max_{u\in\mathbb{R}}S^k_m(u)=\alpha^k_m\;\log\left(1+\tsinr^k_m\right)\quad\mbox{and}\quad u^*=\arg\max_{u\in\mathbb{R}}S^k_m(u)=\log\left(1+\tsinr^k_m\right)+1\quad\forall k,m.
\end{equation}
Therefore, by incorporating the slack variables $\bu=[u^k_m]$ and substituting the objective function of $\R(\bP)$ with its equivalent term $\sum_{m=1}^M\sum_{k=1}^K\max_{u^k_m}S^k_m(u^k_m)$, the problem $\R(\bP)$ is equivalently given by
\begin{equation}\label{eq:sum_rate3}
    \R(\bP)=\left\{
    \begin{array}{ll}
      \max_{\{\bPhi_m\},\bu} & \sum_{m=1}^M\sum_{k=1}^K\;\alpha^k_m\;u^k_m-\frac{\alpha^k_m}{1+{\sf sinr}^k_m}\exp(u^k_m-1) \\
      {\rm s.t.} & \|\bPhi_m\|^2_2\leq P_m\quad\forall  m.
    \end{array}
    \right.
\end{equation}
For any fixed $\bu$ we define the intermediate problem  $\tilde\R(\bP,\bu)$ which yields the optimal precoders $\{\bPhi_m\}$ corresponding to the given $\bu$ and power budget $\bP$. Since for a given $\bu$ the term $\sum_m\sum_k\;\alpha^k_m\;u^k_m$ becomes a constant we get
\begin{equation}\label{eq:R_u}
    \tilde\R(\bP,\bu)\dff\left\{
    \begin{array}{ll}
      \min_{\{\bPhi_m\}} & \sum_{m=1}^M\sum_{k=1}^K\frac{\alpha^k_m}{1+{\sf sinr}^k_m}\exp(u^k_m-1) \\
      {\rm s.t.} & \|\bPhi_m\|^2_2\leq P_m\quad\forall  m.
    \end{array}
    \right.
\end{equation}
The problem $\tilde\R(\bP,\bu)$ can now be transformed into a weighted sum of the worst-case $\mse$s as follows,
\begin{equation}\label{eq:R_u2}
    \tilde\R(\bP,\bu)=\left\{
    \begin{array}{ll}
      \min_{\{\bPhi_m\}} & \sum_{m=1}^M\sum_{k=1}^K\alpha_m^k\exp(u^k_m-1)\max_{\{\bDelta^k_{m,n}\}} \min_{f^k_m}\tmse^k_m \\
      {\rm s.t.} & \|\bPhi_m\|^2_2\leq P_m\quad\forall  m.
    \end{array}
    \right.
\end{equation}
Next, for any given $\bu$ we find an upper bound on $\tilde\R(\bP,\bu)$, which by recalling (\ref{eq:sum_rate3}) and (\ref{eq:R_u2}) is a lower bound on $\R(\bP)$. By invoking the inequality in (\ref{eq:inequality})
and defining $\bbf_m=[f^k_m]_{k}$,  an upper bound on $\tilde\R(\bP,\bu)$ can be found as follows,
\begin{equation}\label{eq:R_u3}
    \bar\R(\bP,\bu)=\left\{
    \begin{array}{ll}
      \min_{\{\bPhi_m,\bbf_m\}} & \sum_{m=1}^M\sum_{k=1}^K\alpha^k_m\exp(u^k_m-1) \;\max_{\{\bDelta^k_{m,n}\}}\tmse^k_m\\
      {\rm s.t.}       & \|\bPhi_m\|^2_2\leq P_m\quad\forall  m.
    \end{array}
    \right.
\end{equation}
$\bar\R(\bP,\bu)$ can itself be sub-optimally solved by using the alternating optimization (AO) principle and optimizing $\{\bbf_m\}$ and $\{\bPhi_m\}$ in an alternating manner. By deploying AO we can  optimize $\{\bPhi_m\}$ while keeping $\{\bbf_m\}$ fixed and vice versa. Since the objective is bounded and it decreases monotonically at each iteration, the AO procedure is guaranteed to converge. In the following theorem we show that solving $\bar\R(\bP,\bu)$ at each step of the AO procedure is a convex problem with a computationally efficient solution.

\begin{theorem}
\label{th:sum_rate}
For arbitrarily fixed $\{\bPhi_m\}$, the problem $\bar\R(\bP,\bu)$ can be optimized over  $\{\bbf_m\}$ efficiently as an SDP. Similarly, for arbitrarily fixed $\{\bbf_m\}$, the problem $\bar\R(\bP,\bu)$ can be optimized over $\{\bPhi_m\}$  efficiently as another SDP.
\end{theorem}
\begin{proof}
See Appendix \ref{app:th:sum_rate}.
\end{proof}

Algorithm 3 summarizes the steps required for sub-optimally solving $\R(\bP)$. This algorithm is constructed based on the connection between the objective functions of $\R(\bP)$ and $\bar\R(\bP,\bu)$. At each iteration of Algorithm 3 for a fixed $\bu$,   $\bar\R(\bP,\bu)$ is solved by using the AO principle as discussed above and a new set of precoders and equalizers is obtained. The minimum rate achieved by using this set of precoders and equalizers provides a lower bound on $\R(\bP)$. This set of precoders and equalizers is also deployed for computing the worst-case $\mse$s and updating $\bu$ as $u^k_m=1-\log\left(\max_{\{\bDelta^k_{m,n}\}}\tilde\mse_m^k\right),\;\forall\;k,m$.\footnote{Note that the worst-case $\mse$s can be computed using the techniques given in Appendix \ref{app:th:GEVP}.}  Since $\R(\bP)$ is bounded from above, so is any lower bound on it. Therefore, the utility function of Algorithm 3 is bounded and increases monotonically in each iteration. Thus, convergence of Algorithm 3 is guaranteed.

\begin{figure*}[t]
{\small \begin{minipage}[t]{6.5 in}
\rule{\linewidth}{0.3mm}\vspace{-.05in}
{\bf Algorithm 3} - Robust Weighted Sum-rate Optimization \vspace{-.15in}\\
\rule{\linewidth}{0.2mm}
{{\rm
\begin{tabular}{ll}
   \;1:& Input $\bP$ and $\{\tilde\bh^k_{m,n},\epsilon_{m,n}^k\}$\\
   \;2:& Initialize $\bPhi_m,\bbf_m$ for all $m$ and $\bu$\\
   \;3:& \textbf{repeat}\\
   \;4:& \quad Solve $\bar\R(\bP,\bu)$  by optimizing over $\{\bPhi_m\}$ and $\{f^k_m\}$ in an alternating manner;  \\
   \;5:& \quad Update $u^k_m\leftarrow 1-\log\left(\max_{\{\bDelta^k_{m,n}\}}\tilde\mse_m^k\right),\;\forall\;k,m$\\
   \;6:& {\bf until} convergence\\
   \;7:& Output $\{\bPhi_m\}$
\end{tabular}}}\\
\rule{\linewidth}{0.3mm}
\end{minipage}}
\end{figure*}

\subsection{Distributed Implementation}
\label{sec:sum_limited}
An advantage of the  AO based approach employed to sub-optimally solve $\R(\bP)$ in the previous section is that it is amenable to a distributed implementation. In particular, note that for fixed  $\bu,\{\bbf_m\}$, the optimization over $\{\bPhi_m\}$ decouples into $M$ smaller problems (\ref{eq:appopP}) that can be solved concurrently by the $M$ BSs. Similarly, for fixed $\bu,\{\bPhi_m\}$, the optimization over $\{\bbf_m\}$ decouples into $KM$ smaller problems (\ref{eq:appopf}) that can be solved concurrently. Finally, for a given $\{\bPhi_m,\bbf_m\}$ the elements of $\bu$ can also be updated concurrently. Consequently,  Algorithm 3 can indeed be implemented in a distributed fashion with appropriate information exchange among the BSs.

\section{High $\snr$ Analysis: Degrees of Freedom}\label{sec:highsnr}

In this section, we analyze the high $\snr$ behavior of the robust max-min rate and the robust weighted sum rate. We suppose that for each user $||\bh^k_{m,m}||>\epsilon_{m,m}^k$ and that the power of the $m^{th}$ BS scales as $\gamma P_m,\;\forall\;m$, where $\gamma\to\infty$.
Then, for any given choice of precoding matrices $\{\bPhi_m\}$, we examine the worst-case sinr given in (\ref{eq:sinr_tilde_MU}).
Notice that by choosing any particular error vectors from their respective uncertainty regions, we can obtain upper-bounds on the worst-case sinrs $\{\tsinr^k_m\}$.
In particular, to obtain an upper-bound on $\tsinr^k_m$, for each interfering BS $n\neq m$, we select a user $q$ and choose
 an error vector  $ \bDelta^k_{m,n}$  that maximizes $|\bh^k_{m,n}\bw^q_n|^2$.
 Similarly, for BS $m$ we select a user $q\neq k$ and choose
 an error vector  $ \bDelta^k_{m,m}$  that maximizes $|\bh^k_{m,m}\bw^q_m|^2$. With these particular
  $ \{\bDelta^k_{m,n}\}$ and some algebra, we obtain an upper bound given by
{\small\begin{equation}\label{eq:sinr_tilde_MU2}
    \frac{(\|\tilde\bh^k_{m,m}\|+\epsilon_{m,m}^k)^2\|\bw^k_m\|^2}
    {1+ f(\tilde\bh^k_{m,m},\bw^q_m,\{\bw^j_m\}_{j\neq q,k},\epsilon_{m,m}^k) +
    \sum_{n\neq m} f(\tilde\bh^k_{m,n},\bw^q_n,\{\bw^j_n\}_{j\neq q},\epsilon_{m,n}^k)},
\end{equation}}
 where $\tilde\bw^q_n=\bw^q_n/\|\bw^q_n\|,\;\forall\;q,n$ and
 {\small\begin{equation}
    f(\tilde\bh^k_{m,n},\bw^q_n,\{\bw^j_n\}_{j\neq q},\epsilon_{m,n}^k)= (|\tilde\bh^k_{m,n}\tilde\bw^q_n|+\epsilon_{m,n}^k)^2\|\bw^q_n\|^2 + \sum_{j\neq q}(|\tilde\bh^k_{m,n}\tilde\bw^j_n|-\epsilon_{m,n}^k|\tilde\bw^{j\;H}_n\tilde\bw^q_n|)^2\|\bw^j_n\|^2.
\end{equation}}
Then, we can optimize over the choice of $q$ to obtain a tighter upper bound  given by
{\small\begin{equation}\label{eq:sinr_tilde_MUub}
   \frac{(\|\tilde\bh^k_{m,m}\|+\epsilon_{m,m}^k)^2\|\bw^k_m\|^2}
    {1+ \max_{q\neq k}\{f(\tilde\bh^k_{m,m},\bw^q_m,\{\bw^j_m\}_{j\neq q,k},\epsilon_{m,m}^k)\}+
    \sum_{n\neq m} \max_{q}\{f(\bh^k_{m,n},\bw^q_n,\{\bw^j_n\}_{j\neq q},\epsilon_{m,n}^k)\}}.
\end{equation}}
 Note that the upper bound in (\ref{eq:sinr_tilde_MUub}) is tight when we have perfect CSI, i.e., when all $\epsilon_{m,n}^k=0$. A simpler upper bound that suffices for a result proved in the sequel is given below.
 {\small\begin{equation}\label{eq:sinr_tilde_MUub2}
 \tsinr^{k\;{\rm ub}}_m =  \frac{(\|\tilde\bh^k_{m,m}\|+\epsilon_{m,m}^k)^2\|\bw^k_m\|^2}
    {1+ \max_{q\neq k}\{(|\tilde\bh^k_{m,m}\tilde\bw^q_m|+\epsilon_{m,m}^k)^2\|\bw^q_m\|^2\}+
    \sum_{n\neq m} \max_{q}\{(|\tilde\bh^k_{m,n}\tilde\bw^q_n|+\epsilon_{m,n}^k)^2\|\bw^q_n\|^2\}}.
\end{equation}}
 The following theorem characterizes the high $\snr$ behavior of the robust max-min rate and the robust weighted sum rate.
 \begin{theorem}\label{thm:highsnr}
Suppose that $\epsilon_{m,n}^k>0,\;\forall\;k,m,n$ and the vector of transmit powers is scaled as $\gamma\bP$. Then, the minimum worst-case  $\sinr$ $\CS(\gamma\bP)$ as well as the worst-case weighted sum rate $\R(\gamma\bP)$ are monotonically increasing in $\gamma$. Moreover, we have that
\begin{equation}\label{eq:highsnr}
 \lim\sup_{\gamma\to\infty} \CS(\gamma\bP)<\infty\;\;\;\;\&\;\;\;\;\lim_{\gamma\to\infty} \frac{\R(\gamma\bP)}{\log(\gamma)}=\max_{k,m}\{\alpha^k_m\}.
\end{equation}
\end{theorem}
 \proof First, we note that for any given vectors $\{\bw^k_m\}$ and any choice of channel vectors, we have that
 \begin{equation}
   \frac{\gamma|\bh^k_{m,m}\bw^k_m|^2}{\sum_{l\neq k}\gamma|\bh^k_{m,m}\bw^l_m|^2 +\sum_{n\neq m}\sum_l\gamma|\bh^k_{m,n}\bw^l_n|^2+1}
\end{equation}
is monotonically increasing in $\gamma>0$. Consequently, we have that each of the worst-case sinr given in (\ref{eq:sinr_tilde_MU}) must increase when the beamforming vectors are scaled as $\{\sqrt{\gamma}\bw^k_m\}$. Thus, the minimum worst-case  sinr $\CS(\gamma\bP)$ as well as the worst-case weighted sum rate $\R(\gamma\bP)$ are monotonically increasing in $\gamma$.
 Next, we note that $\tsinr^{k\;{\rm ub}}_m$ in (\ref{eq:sinr_tilde_MUub2}) can further be upper bounded as
  {\small\begin{equation}\label{eq:sinr_tilde_MUub3}
 \tsinr^{k\;{\rm ub}}_m \leq \frac{(\|\tilde\bh^k_{m,m}\|+\epsilon_{m,m}^k)^2\|\bw^k_m\|^2}
    {1+ \sum_{q\neq k}(\epsilon_{m,m}^k)^2\|\bw^q_m\|^2/(K-1)+
    \sum_{n\neq m} \sum_{q}(\epsilon_{m,n}^k)^2\|\bw^q_n\|^2/K}.
\end{equation}}
The RHS in (\ref{eq:sinr_tilde_MUub3}) corresponds to the SINR in a fully connected $MK-$user Gaussian interference channel (GIC) with constant channel gains and where each source and each receiver have a single antenna. For this GIC it is known that the symmetric rate saturates as the transmit power of each source is simultaneously increased and that the total degrees of freedom is one \cite{Tjaf:IT}. The results in (\ref{eq:highsnr}) now follows.  \endproof
The result in Theorem \ref{thm:highsnr} must be contrasted to the interference-alignment results that have recently been developed (cf. \cite{Tjaf:IT}).  We recall that the key idea behind interference alignment is to ensure that the all the interference seen by a user (receiver) aligns itself (i.e.,  confines itself to a sub-space) so that the user sees the remaining subspace as interference free. Such an alignment based solution possesses a first-order optimality at high $\snr$ in the case of perfect CSIT.  Theorem \ref{thm:highsnr} shows that such an alignment is not possible over our model due to the uncertainty involved in the available CSIT. Thus unlike the perfect CSIT case no simple high-$\snr$ alignment solutions are possible and consequently  algorithms that consider robust optimization at finite $\snr$s remain the only viable options.

Next, we consider another good albeit sub-optimal approach for selecting the beamforming vectors.
We consider the worst-case signal-to-leakage-interference-plus-noise ratio (SLINR). The SLINR metric has been shown to be an effective metric over networks with perfect CSI \cite{TarighatA:ICASSP05,Lvent:Tw10}.  In particular, the worst-case SLINR corresponding to user $k$ in cell $m$ is given by,
\begin{equation}\label{eq:slinr}
 \tslinr^k_m\dff\min_{\{\bDelta^k_{m,n}\}}\frac{|\bh^k_{m,m}\bw^k_m|^2}{\sum_{j\neq k}|\bh^j_{m,m}\bw^k_m|^2 +\sum_{n\neq m}\sum_l|\bh^l_{n,m}\bw^k_m|^2+1}
\end{equation}
Notice that the uncertainty regions in the numerator and denominator are decoupled so that
\begin{equation}\label{eq:slinr2}
 \tslinr^k_m = \frac{\min_{\{\bDelta^k_{m,m}\}}|\bh^k_{m,m}\bw^k_m|^2}{\sum_{j\neq k}\max_{\{\bDelta^j_{m,m}\}}|\bh^j_{m,m}\bw^k_m|^2 +\sum_{n\neq m}\sum_l\max_{\{\bDelta^l_{n,m}\}}|\bh^l_{n,m}\bw^k_m|^2+1}
\end{equation}
Suppose that a per-user power profile $\{P_m^k\}$ has been given. Then, the beamforming vectors can be independently designed by solving
\begin{equation}\label{eq:slinr4}
     \left\{
    \begin{array}{ll}
      \max_{\bw^k_m} & \tslinr^k_m \\
      {\rm s.t.} & \|\bw^k_m\|^2_2\leq P^k_m\quad\forall \; \;k,m.
    \end{array}
    \right.
\end{equation}
The maximization problem in (\ref{eq:slinr4}) can be exactly solved
by alternatively solving a power optimization problem in conjunction with a linear bi-section search as described in
 (\ref{eq:rate3}) and (\ref{eq:f2}).

\section{Simulation Results}
\label{sec:simulations}
In this section, we provide simulation results to assess the
performance of the proposed algorithms. Both robust max-min rate  and robust weighted sum-rate problems as discussed in Sections \ref{sec:maxmin} and \ref{sec:sum} can be posed or conservatively approximated by semidefinite programs. Therefore, we use the software package SDPT3 developed as a MATLAB toolbox for solving semidefinite programs \cite{SDTP3}.

\begin{figure}
  \centering
  \includegraphics[width=4.5 in]{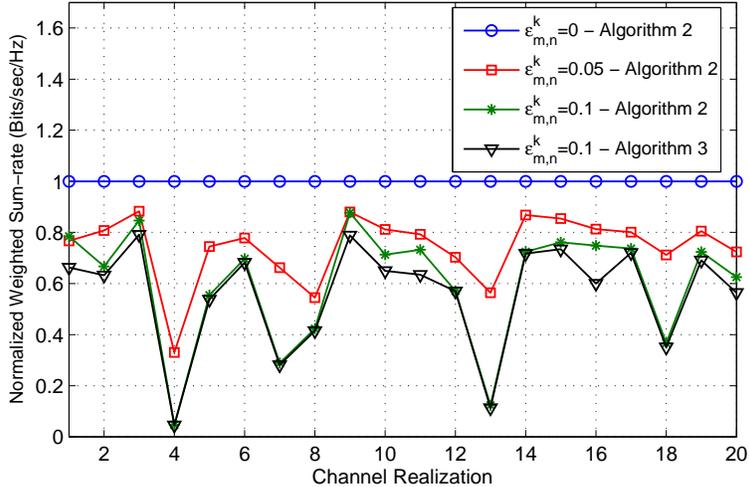}\\
  \caption{Comparing the robust max-min rates obtained with centralized and distributed precoder designs over 20 independent channel realizations and different uncertainty regions.}\label{fig:maxmin_real}
\end{figure}

We consider a network consisting of two cells ($M=2$) and each cell with two users ($K=2$). Each BS is equipped with two transmit antennas ($N=2$) and each user has one receive antenna. Also we consider identical power constraints for all BSs and set $P_m=10$ dB for $m=1,2$. Fig. \ref{fig:maxmin_real} compares the optimized minimum worst-case  rates  (robust max-min rates) achieved by  full cooperation (Algorithm 1) and limited cooperation (Algorithm 2), respectively, among the BSs. Note that the performance yielded by Algorithm 1 can also be achieved in a more distributed manner as described in Section \ref{sec:maxmin_limited}. We consider 20 channel realizations and for each realization we solve the robust max-min rate problem under four different setups. As the baseline we consider the fully cooperative scenario with {\em perfect} CSI, i.e., $\epsilon^k_{m,n}=0$ for all $k,m,n$. For perfect CSI, the $\sinr$ lower bound  given in \eqref{eq:sinr_bar} becomes the exact $\sinr$, i.e., $\bsinr^k_m=\tsinr^k_m=\sinr^k_m$ and therefore the solution of Algorithm 1  is the {\em optimal} solution. We normalize the robust max-min rates obtained in different scenarios by this optimal robust max-min rate. Next, we obtain the robust max-min rates for different uncertainty regions with radii $\epsilon^k_{m,n}=0.05$ and $0.1$ for all $k,m,n$. It is observed that larger uncertainty regions result in smaller robust max-min rates, which is expected. Finally, we assess the robust max-min rate obtained by the distributed algorithm with limited cooperation (Algorithm 3), where each BS updates its precoder unilaterally and the radii of the uncertainty regions are  $\epsilon^k_{m,n}=0.1$ for all $k,m,n$. For some channel realizations, the solution obtained by this distributed algorithm is precisely equal to that of the algorithm with full cooperation. For most realizations, however, Algorithm 2 exhibits degraded performance compared to Algorithm 1 . This is the cost incurred for the benefit of having limited cooperation between the BSs. In Fig. \ref{fig:maxmin_real_comp} we consider the same setup as in Fig.~\ref{fig:maxmin_real} and compare the relative performance yielded by the two distributed algorithms which solve the robust max-min rate optimization problem through power optimization and $\mse$ optimization. According to the simulation results, neither of these two algorithms consistently outperforms the other one. Also, it is observed that for any channel realization, the performance of the better one is almost close to that of the situation when the perfect CSI is viable.

\begin{figure}
  \centering
  \includegraphics[width=4.5 in]{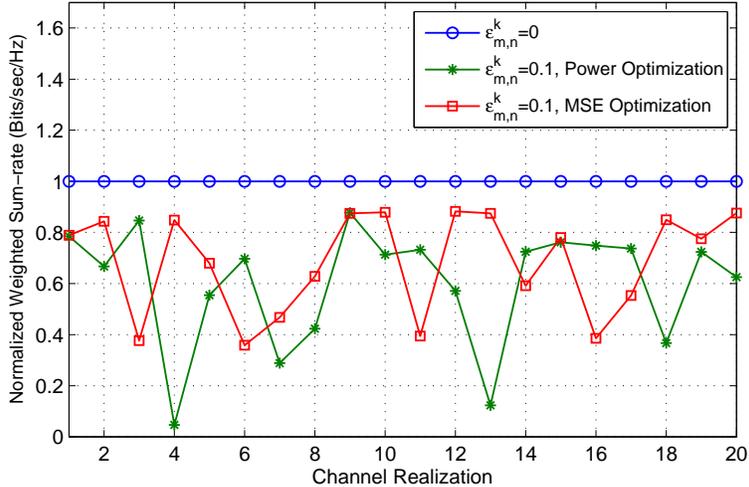}\\
  \caption{Comparing the robust max-min rates obtained with the two proposed distributed algorithms (power optimization and $\mse$ optimization) over 20 independent channel realizations and different uncertainty regions.}\label{fig:maxmin_real_comp}
\end{figure}
\begin{figure}
  \centering
  \includegraphics[width=4.5 in]{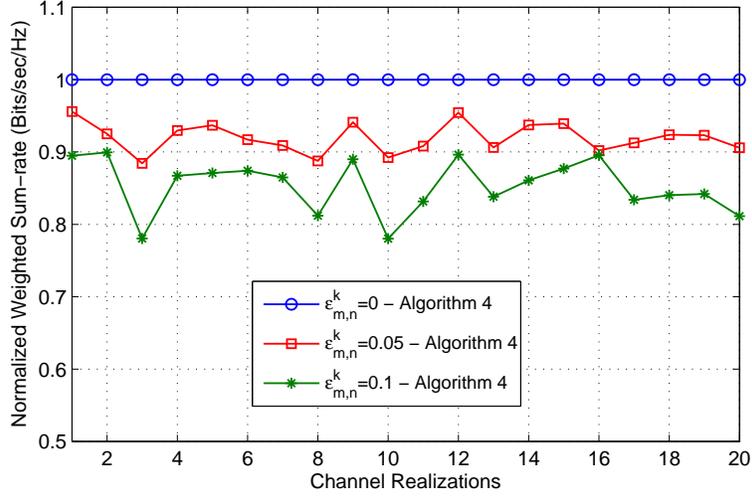}\\
  \caption{Comparing the robust weighted sum-rates obtained  over 20 independent channel realizations and different uncertainty regions.}\label{fig:sum_real}
\end{figure}
\begin{figure}
  \centering
  \includegraphics[width=4.5 in]{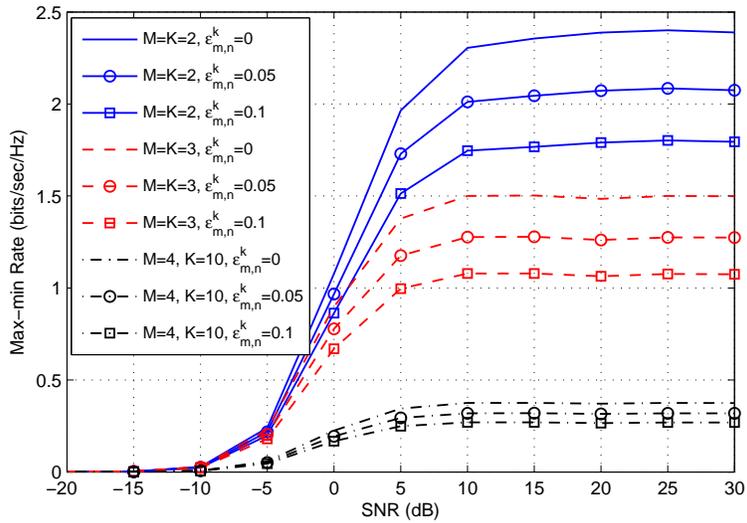}\\
  \caption{The robust max-min rate versus $\snr$ (dB) for different levels of uncertainty.}\label{fig:R_maxmin}
\end{figure}

In Fig. \ref{fig:sum_real} we consider the same network setup and also for convenience set the rate weighting factors equal to 1, i.e., $\alpha^k_m=1$ for all $k,m$. Similar to Fig. \ref{fig:maxmin_real}, we examine the   uncertainty regions with radii $\epsilon^k_{m,n}=0, \;0.05$ and $0.1$, respectively  and plot the achievable robust weighted sum-rates which are determined using Algorithm 3. The relative performance of different settings are identical to those obtained  in Fig. \ref{fig:maxmin_real}. The key observation is that the performance degradation due to larger uncertainty regions  is smaller for the robust weighted sum-rate problem than that seen in the robust max-min rate case since the former is less vulnerable to the undesired CSI noise and hence is expected to degrade more gracefully as the uncertainty regions expand.

\begin{figure}
  \centering
  \includegraphics[width=4.5 in]{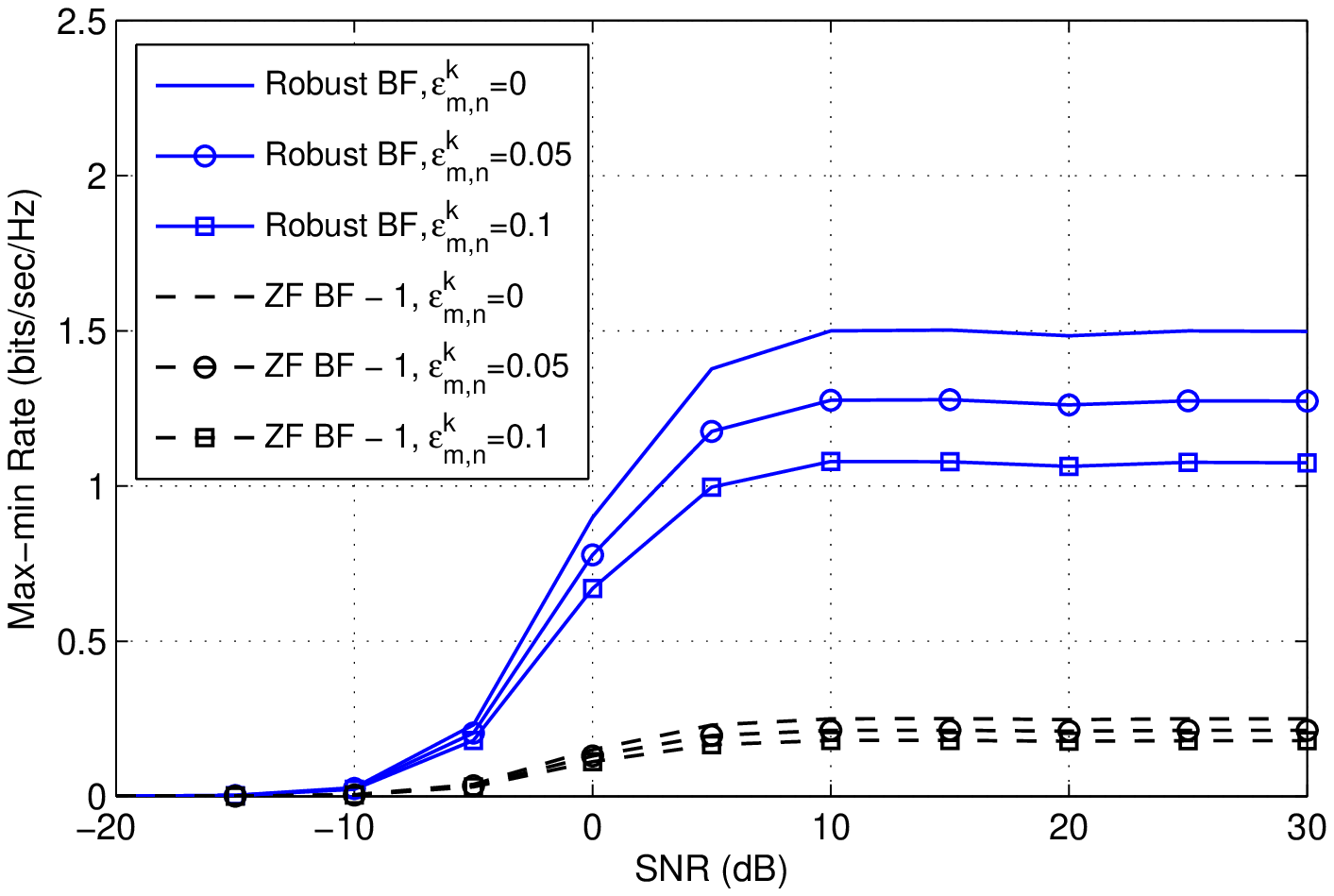}\\
  \caption{Comparing the minimum worst-case rates yielded by the robust beamforming and zero-forcing beamforming designs for $M=N=K=3$.}\label{fig:ZF_maxmin}
\end{figure}

Fig.~\ref{fig:R_maxmin} plots the robust max-min rate (achieved using Algorithm 1) versus $\snr$. Here, we consider three different network settings; one with three cells ($M=3$) each with three users ($K=3$), one with two cells ($M=2$) each with two users ($K=2$), and finally one with four cells ($M=4$) each with $(K=10$) users. In the first two settings we assume that the number of transmit antennas per BS are $N=3$ and in the third one we assume $N=4$ transmit antennas per BS. It is observed that for each fixed uncertainty region, there exists a considerable gap between the robust max-min rates of the settings $M=K=3$ and $M=K=2$ as well as between that of the settings $M=10, K=10$ and $M=K=3$. This is due to the fact that independent messages are transmitted to different users, and therefore increasing the number of users from 4 to 40 increases the amount of the interference imposed on each user, which in turn degrades the quality of the communication for all users. Moreover, as the number of users increases, the likelihood that the weakest user suffers from a very weak communication quality increases. Further, as predicted by Theorem \ref{thm:highsnr}, for each setting the robust max-min rate  saturates at high $\snr$. We note that the saturation of the optimized min rate at high $\snr$ even in case of perfect CSI can be deduced from \cite{Tjaf:IT}. In particular, we can infer from the results in \cite{Tjaf:IT} that assigning one degree of freedom to each user is not possible for any of these three settings.\footnote{Note that assigning a fractional degree of freedom to any user is not possible with our model since we do not allow precoding (beamforming) across multiple time and/or frequency slots.} Moreover, from Fig. \ref{fig:ZF_maxmin} (which considers $M=K=N=3$) we see that the optimized robust designs yield a substantial improvement in the minimum worst-case rate compared to the naive zero-forcing strategy, wherein each BS designs beam vectors for its in-cell users under the assumption that it alone operates in the network and that the channel estimate vectors available to it are perfect. Each BS performs the zero-forcing operation on   its channel estimate vectors and  then does power allocation to maximize the minimum rate among its in-cell users.

\begin{figure}
  \centering
  \includegraphics[width=4.5 in]{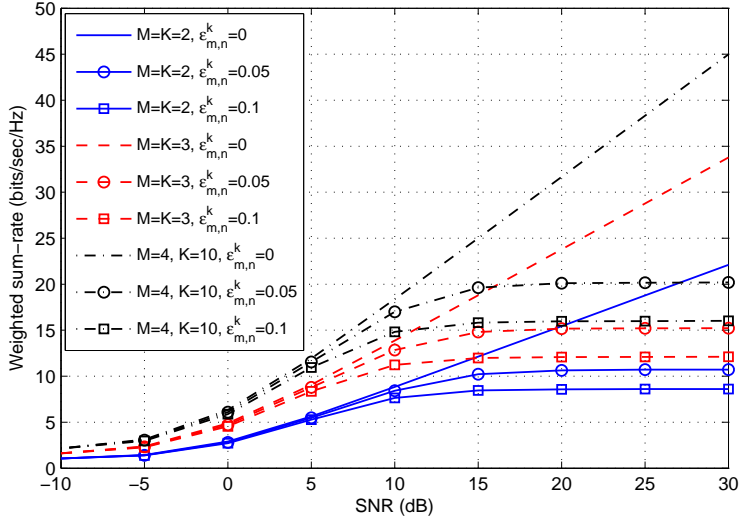}\\
  \caption{The robust sum-rate versus $\snr$ (dB) for different levels of uncertainty.}\label{fig:R_sum}
\end{figure}

Fig. \ref{fig:R_sum} depicts the optimized worst-case sum-rate (robust sum-rate)   for the same network settings and uncertainty regions as in Fig.~\ref{fig:R_maxmin}.  Note that unlike  the robust max-min rate,    the robust sum-rate must increase with the network size.  Also,  for each setting the optimized worst-case sum rate  saturates at high $\snr$ while Theorem \ref{thm:highsnr} predicts that the robust sum-rate has one degree of freedom, i.e., scales as $\log(\snr)$. This is due to the simple albeit sub-optimal AO technique employed wherein conservative bounds are optimized at each step. However, with perfect CSI we obtain a positive total degrees of freedom. For the model at hand with $M=K=3$ ($M=K=2$), using the results in \cite{Tjaf:IT}  an upper bound on the total degrees of freedom can be computed to be $4.5$ ($3$), after assuming perfect in-cell user cooperation and allowing for precoding across multiple time and/or frequency slots. This upper-bound needs not be achievable and the total degrees of freedom we observe from the plot is $3$ ($2$). Similarly we observe that for the setting $M=4, K=10$ the number of degrees of freedom is 4. Next, in Fig. \ref{fig:ZF_sum2} we consider the setting $M=K=3$ and compare the worst-case sum-rates yielded by the robust designs, the naive zero-forcing strategy and the SLINR beamforming. Note that in the latter two cases each BS computes its beamforming vectors independently. In particular,  the beam vectors were obtained in  the zero-forcing case as described for the example in Fig. \ref{fig:ZF_maxmin}, except that the power allocation is done to maximize the sum-rate. On the other hand,  the beam vectors were obtained in the SLINR case as described in Section \ref{sec:highsnr} for each power profile and an exhaustive search was conducted over power profiles to maximize the sum SLINR. Note that gains obtained by the robust and SLINR-based designs over the zero-forcing design are significant and commensurate with the extent to which they account for the  interference seen from other sources and that imposed on other users.
\begin{figure}
  \centering
  \includegraphics[width=4.5 in]{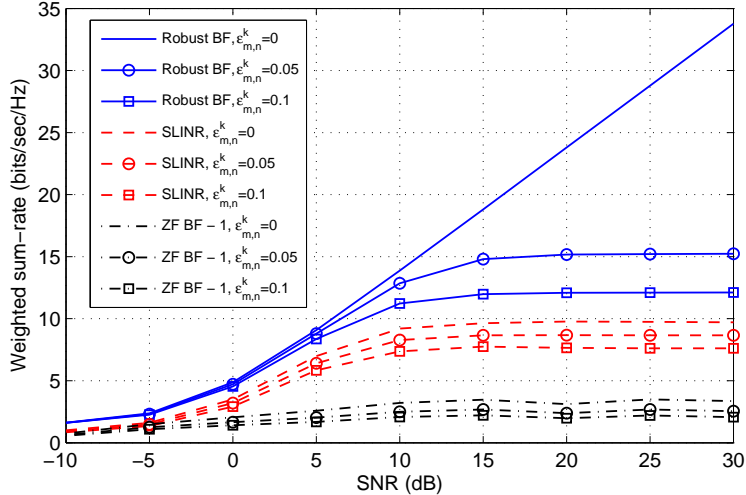}\\
  \caption{Comparing the worst-case weighted sum-rates yielded by the robust beamforming, zero-forcing, and {\sf SLINR}-based beamforming designs for $M=N=K=3$.}\label{fig:ZF_sum2}
\end{figure}

\section{Conclusions}
\label{sec:conclusions}
We have considered designing robust precoders for multi-cell multiuser downlink systems when the BSs can acquire only noisy channel estimates. To account for the uncertainty about the channel states we adopt the notion of worst-case robustness and aim at maximizing the network minimum rate and a weighted sum-rate of all users for the worst-case estimation perturbation. Depending on the level of cooperation among the BSs, algorithms with full cooperation  and limited cooperation  are offered for designing the precoders. The precoder design problems can be either posed as convex problems, or conservatively approximated by some convex problems. All convex problems are shown to have computationally  efficient solutions. Table~\ref{table} summarizes the main proposed algorithms.

{\scriptsize
\begin{table}[h]
   \hspace{0.1 in}
  \begin{minipage}{\textwidth}
  {\small
  \begin{tabular}{|l|l|l|l|}
    \hline
    Problem & Optimality& Complexity & Type  \\ \hline\hline
    Robust max-min rate ($K=1$)& optimal  & SDP
    &  centralized\\
    Robust max-min rate via power optimization ($K>1$)& suboptimal & SDP   & centralized  \\
    Robust max-min rate via MSE optimization ($K>1$)& suboptimal  & GEVP   & centralized \\
    Robust max-min rate via power optimization ($K>1$)& suboptimal  & SDP   & distributed\\
    Robust max-min rate via power optimization ($K>1$)& suboptimal  & SDP   & distributed\\
    Robust weighted sum-rate ($K>1$)& suboptimal & SDP   & distributed\\
    \hline
  \end{tabular}}
  \end{minipage}
  \caption{Summary of the  algorithms.}\label{table}
  \footnotetext[1]{table footnote }
\end{table}
}
\appendix

\section{Proof of Theorem \ref{th:connection}}
\label{app:th:connection}
Let us denote the set of precoders obtained from solving $\CS(\bP)$  by $\{\bw^*_m\}$ and their corresponding worst-case $\sinr$s by $\tsinr^*_m$. From the definition of $\CS(\bP)$ we have
\begin{equation*}
    \frac{\|\bw^*_m\|^2_2}{P_m}\leq 1\quad\forall m,\quad\mbox{and}\quad\min_m\left\{\tsinr^*_m\right\}=\; \CS(\bP)\quad\Rightarrow \quad \tsinr^*_m\geq \CS(\bP)\quad\forall m.
\end{equation*}
From the definition of $\CP(\bP,a)$ we find that for the choice of $\{\bw^*_m\}$, the choice of $b=1$ is achievable for $\CP(\bP,\CS(\bP))$ and therefore $\CP(\bP,\CS(\bP))\leq 1$.

Next we show that $\CP(\bP,\CS(\bP))$ cannot be less than one. Let us denote the set of precoders obtained by solving $\CP(\bP,\CS(\bP))$ by $\{\bw^{**}_m\}$. From the definition of $\CP(\bP,\CS(\bP))$ we clearly have $\tsinr^{**}_m\geq \CS(\bP)$. If $\CP(\bP,\CS(\bP))<1$ i.e., if $\max_m\frac{\|\bw^{**}_m\|^2_2}{P_m}=c<1$, then we define the set of precoders   $\{\hat\bw_m\}\dff\left\{\bw^{**}_m/\sqrt{c}\right\}$. $\{\hat\bw_m\}$ clearly satisfy the power constraints and moreover for their corresponding worst-case $\sinr$s from (\ref{eq:sinr_tilde}) we have
\begin{equation*}
    \hat\tsinr_m=\frac{\frac{1}{c}\big|(|\tilde\bh_{m,m}\bw^{**}_m|-\epsilon_{m,m}\|\bw^{**}_m\|_2)^+\big|^2} {\frac{1}{c}\sum_{n\neq m}\big||\tilde\bh_{m,n}\bw^{**}_n|+\epsilon_{m,n}\|\bw^{**}_n\|_2\big|^2+1}> \frac{\big|(|\tilde\bh_{m,m}\bw^{**}_m|-\epsilon_{m,m}\|\bw^{**}_m\|_2)^+\big|^2} {\sum_{n\neq m}\big||\tilde\bh_{m,n}\bw^{**}_n|+\epsilon_{m,n}\|\bw^{**}_n\|_2\big|^2+1},
\end{equation*}
since $c<1$. Therefore, we have found a set of precoders $\{\hat\bw_m\}$ which satisfy the power constraints and yet yield a strictly larger robust max-min $\sinr$ compared to what the precoders $\{\bw^{*}_m\}$ obtain. This contradicts the optimality of $\{\bw^{*}_m\}$ and therefore $\CP(\bP,\CS(\bP))=1$. The strict monotonicity and continuity of $\CP(\bP,a)$  in $a$, at any strictly feasible $a$, follows from a similar line of argument.

\section{Proof of Theorem \ref{th:SU1}}\label{app:th:SU1}
By considering the characterization of $\tsinr_m$ given in (\ref{eq:sinr_tilde}), the constraint $\tsinr_m\geq a$ provides that $\forall m,\;\exists\; t_m\in\mathbb{R^+}$ such that
\begin{eqnarray}
    \label{eq:break1} \quad
    \big|(|\tilde\bh_{m,m}\bw_m|-\epsilon_{m,m}\|\bw_m\|)^+\big|^2&\geq& at^2_m,\\
    \label{eq:break2}\mbox{and} \quad
    \sum_{n\neq m} \big||\tilde\bh_{m,n}\bw_n|+\epsilon_{m,n}\|\bw_n\|\big|^2+1 &\leq& t^2_m.
\end{eqnarray}
Next, note that the optimal solutions are insensitive to any phase shift. In other words, if $\bw^*_m$ is an optimal solution, then $\bw^*_me^{j\theta}$ is also an optimal solution as such phase shifts do not alter the objective or the constraints of $\CP(\bP,a)$ given in (\ref{eq:f}). Among such optimal solutions we select those for which $\tilde\bh_{m,m}\bw_m,\forall\;m$ has a non-negative real part and a zero imaginary part. Therefore, (\ref{eq:break1}) can be restated as
\begin{equation}
   \label{eq:break3}
   \epsilon_{m,m}\|\bw_m\|\leq\tilde\bh_{m,m}\bw_m-\sqrt{a}t_m,
\end{equation}
which is a second-order cone (SOC) constraint. In this context, we note that the useful step in (\ref{eq:break3}) was first developed in \cite{Gershman:TSP03}, wherein an inequality of the form $|\tilde\bh\bw|-\epsilon\|\bw\|\geq 1$ was expressed as a convex constraint.

Furthermore, by introducing the additional slack variables $\{c_{m,n}\}$, $\{d_{m,n}\}$, and $\{e_{m,n}\}$, corresponding to the terms $|\tilde\bh_{m,n}\bw_n|$, $\epsilon_{m,n}\|\bw_n\|$, and $|\tilde\bh_{m,n}\bw_n|+\epsilon_{m,n}\|\bw_n\|$, respectively, we can express the constraint in (\ref{eq:break2}) equivalently by
\begin{eqnarray}
    \label{eq:break4}
    \left\{
    \begin{array}{ll}
    \sqrt{\sum_{n\neq m}e^2_{m,n}+1} \leq  t_m & \forall m\\
    c_{m,n}+d_{m,n}\leq e_{m,n} & \forall n\neq m\\
    |\tilde\bh_{m,n}\bw_n| \leq c_{m,n} & \forall n\neq m\\
    \epsilon_{m,n}\|\bw_n\|\leq d_{m,n} & \forall n\neq m
    \end{array}\right.,
\end{eqnarray}
which are all SOC or linear constraints.
By defining $\bt=[t_1,\dots,t_M]$, $\bc=[c_{1,1}\dots,c_{M,M}]$, $\bd=[d_{1,1}\dots,d_{M,M}]$, and $\be=[e_{1,1}\dots,e_{M,M}]$ we can reformat $\CP(\bP,a)$ as follows.
\begin{equation}\label{eq:fApp}
    \left\{
    \begin{array}{ll}
      \min_{\{\bw_m\},\bc,\bd,\be, \bt,b} & b \\
      {\rm s.t.} &\epsilon_{m,m}\|\bw_m\|\leq\tilde\bh_{m,m}\bw_m-\sqrt{a}t_m,\\
      &\sqrt{\sum_{n\neq m}e^2_{m,n}+1} \leq  t_m  \quad\forall m\\
    & c_{m,n}+d_{m,n}\leq e_{m,n}  \\
    & |\tilde\bh_{m,n}\bw_n| \leq c_{m,n}  \\
    & \epsilon_{m,n}\|\bw_n\|\leq d_{m,n}  \quad\forall n\neq m\\
    & \frac{\|\bw_m\|}{\sqrt{P_m}}\leq b \quad\forall m
    \end{array}
    \right.
\end{equation}
Note that all the constraints above are linear or second-order cones and the objective is linear in $b$. Therefore, $\CP(\bP,a)$ is an SOC program and hence can also be expressed as an SDP.

\section{Proof of Theorem \ref{th:sdp_m}}
\label{app:th:sdp_m}

Similar to the proof of Theorem \ref{th:SU1}, the constraint $\bsinr^k_m\geq a$ provides that $\forall m,k,\;\exists\; t^k_m\in\mathbb{R^+}$ such that
\begin{eqnarray}
    \label{eq:break1_m} \quad
    a(t^k_m)^2&\leq&\big|(|\tilde\bh_{m,m}^k\bw_m^k|-\epsilon_{m,m}^k\|\bw_m\|_2)^+\big|^2,\\
    \label{eq:break2_m}\mbox{and} \quad
    (t^k_m)^2&\geq& \max_{\bDelta^k_{m,m}} \bh^k_{m,m}\bPsi_{m,k}(\bPsi_{m,k})^H(\bh^k_{m,m})^H +
    \sum_{n\neq m}\max_{\bDelta^k_{m,n}} \bh^k_{m,n}\bPhi_{n}(\bPhi_{n})^H(\bh^k_{m,n})^H+1.
\end{eqnarray}
Next, note that the optimal solutions are insensitive to any phase shift. In other words, if $\bPhi^*_m$ is an optimal solution, then $\bPhi^*_m\times {\rm diag}(e^{j\theta_m^1},\dots, e^{j\theta_m^K})$ is also an optimal solution as such phase shifts do not alter the objective and the constraints of $\CP_1(\bP,\ba)$. Among such optimal solutions we select those for which $\tilde\bh^k_{m,m}\bw^k_m$ has a non-negative real part and a zero imaginary part. Therefore (\ref{eq:break1_m}) for $t^k_m>0$ can be stated as
\begin{equation*}
   \epsilon_{m,m}^k\|\bw^k_m\|_2\leq\tilde\bh^k_{m,m}\bw^k_m-\sqrt{a}t^k_m,
\end{equation*}
which for a given $a$ is a second-order cone (SOC) constraint. Next, by introducing the additional slack variables $\{e_{m,n}^{k}\}$, the other constraints in
 (\ref{eq:break2_m}) can be written as
 \begin{eqnarray}\label{eq:Newcon}
 \nonumber(t^k_m)^2&\geq& (e_{m,m}^k)^2 + \sum_{n\neq m}(e_{m,n}^k)^2 +1,\\
  \mbox{and}\quad \nonumber  (e_{m,m}^k)^2 &\geq&   \bh^k_{m,m}\bPsi_{m,k}(\bPsi_{m,k})^H(\bh^k_{m,m})^H,\;\forall\;
     \bDelta^k_{m,m}:\|\bDelta^k_{m,m}\|\leq \|\epsilon_{m,m}^k\|, \\
    \mbox{and}\quad (e_{m,n}^k)^2 &\geq&   \bh^k_{m,n}\bPhi_{n}(\bPhi_{n})^H(\bh^k_{m,n})^H,\;\forall\;
     \bDelta^k_{m,n}:\|\bDelta^k_{m,n}\|\leq \|\epsilon_{m,n}^k\|.
\end{eqnarray}
We now demonstrate that the constraints in (\ref{eq:Newcon})   can be transformed into finitely many linear matrix inequalities. By applying the Schur Complement lemma \cite{Horn:book}, the constraints
$$\bh^k_{m,n}\bPhi_{n}(\bPhi_{n})^H(\bh^k_{m,n})^H\leq (e^{k}_{m,n})^2\quad\forall \|\bDelta^k_{m,n}\|\leq\epsilon^k_{m,n}$$ can be equivalently stated as
\begin{equation*}
   \left[
    \begin{array}{ccc}
      e^{k}_{m,n} & (\tilde\bh^k_{m,n}+\bDelta^k_{m,n})\bPhi_n \\
      (\bPhi_n)^H(\tilde\bh^{k}_{m,n} +\bDelta^{k}_{m,n})^H  & e^{k}_{m,n}\bI
    \end{array}\right]\succeq 0,\;\;\forall \; \|\bDelta^k_{m,n}\|\leq\epsilon^k_{m,n}.
\end{equation*}
Next, the following lemma proved in \cite{Eldar:SP04} (also used in \cite{Vucic:SP09}) is instrumental for transforming the constraints above into finitely many linear matrix inequalities which account for the uncertainty regions by deploying the additional slack variables $\blambda=[\lambda_{m,n}^k]_{k,m,n}$.
\begin{lemma}
\label{lemma:S}
For any given matrices $\bA$, $\bB$, and $\bC$ with $\bA=\bA^H$, the inequality
\begin{equation*}
    \bA\succeq \bB^H\bD\bC+\bC^H\bD^H\bB\quad \forall \bD:\;\|\bD\|\leq \epsilon,
\end{equation*}
\end{lemma}
holds if and only if
\begin{equation*}
    \exists \lambda\geq 0\quad\mbox{such that}\quad \left[
    \begin{array}{cc}
      \bA-\lambda\bB^H\bB & -\epsilon\;\bC^H \\
      -\epsilon\;\bC & \lambda\bI
    \end{array}\right]\succeq\boldsymbol{0}.
\end{equation*}
By setting $\bB\dff-[1\;\; \boldsymbol{0}]$, $\bC\dff[\boldsymbol{0}\;\;\bPhi_n]$, $\bD=\bDelta^k_{m,n}$, and
\begin{equation*}
   \bA\dff\left[
    \begin{array}{ccc}
      e^{k}_{m,n} & \tilde\bh^k_{m,n}\bPhi_n  \\
      (\bPhi_n)^H(\tilde\bh^k_{m,n})^H & e^{k}_{m,n}\bI
    \end{array}\right],
\end{equation*}
and applying Lemma \ref{lemma:S} we find that the constraints $\bh^k_{m,n}\bPhi_{n}(\bPhi_{n})^H(\bh^k_{m,n})^H\leq (e^{k}_{m,n})^2$ for all $\|\bDelta^k_{m,n}\|\leq\epsilon^k_{m,n}$ are equivalently given by
\begin{equation}\label{eq:LMI}
     \bT^k_{m,n}\dff\left[
    \begin{array}{cccc}
      e^{k}_{m,n} -\lambda_{m,n}^k& \tilde\bh^k_{m,n}\bPhi_n & \boldsymbol{0}\\
      (\bPhi_n)^H(\tilde\bh^k_{m,n})^H & e^{k}_{m,n}\bI & -\epsilon^k_{m,n}(\bPhi_n)^H\\
      \boldsymbol{0} & -\epsilon^k_{m,n}\bPhi_n  & \lambda^{k}_{m,n}\bI
    \end{array}\right]\succeq 0\quad\forall k, m\neq n,
\end{equation}
which is a linear matrix inequality (LMI). Similarly we can show that the constraints $\|\bh^k_{m,m}\bPsi_{m,k}\|_2^2\leq (e^{k}_{m,m})^2$ for all $\|\bDelta^k_{m,m}\|\leq\epsilon^k_{m,m}$ are equivalently given by
\begin{equation}\label{eq:LMI2}
     \bU^{k}_{m,m}\dff \left[
    \begin{array}{cccc}
      e^{k}_{m,m} -\lambda_{m,m}^k & \tilde\bh^k_{m,m}\bPsi_{m,k}& \boldsymbol{0}\\
      (\bPsi_{m,k})^H(\tilde\bh^k_{m,m})^H & e^{k }_{m,m}\bI & -\epsilon^k_{m,m}(\bPsi_{m,k})^H\\
      \boldsymbol{0} & -\epsilon^k_{m,m}\bPsi_{m,k}  & \lambda^{k}_{m,m}\bI
    \end{array}\right]\succeq 0\quad\forall m,k.
\end{equation}

By defining  $\bt=[t_{m}^k]_{k,m}$ and $\be=[e^{k}_{m,n}]$, $\CP_1(\bP,a)$ can be cast as follows.
\begin{equation*}
    \CP_1(\bP,a)=\left\{
    \begin{array}{ll}
      \min_{\{\bPhi_m\},\blambda,\bt,\be,b} & b \\
      {\rm s.t.}
      & \epsilon_{m,m}^k\|\bw^k_m\|_2\leq\tilde\bh^k_{m,m}\bw^k_m-\sqrt{a}t^k_m\quad\forall m, k\\
      & \sqrt{\sum_{n\neq m}(e^{k}_{m,n})^2+(e_{m,m}^{k})^2+1} \leq  t^k_m  \quad\forall m, k\\
      & \bT^k_{m,n}\succeq \boldsymbol{0} \quad \forall \;k, m\neq n,\\
      & \bU^k_{m,m}\succeq \boldsymbol{0} \quad \forall \;k, m,\\
      & \|\bPhi_m\|_2\leq b\sqrt{P_m} \quad\forall m.
    \end{array}\right.,
\end{equation*}
which has linear objective and semidefinite or second-order cones and therefore is an SDP.
\section{Proof of Theorem \ref{th:rate_m_mse}}
\label{app:th:GEVP}
  By recalling (\ref{eq:tmse}) and further defining the slack variables $\{b^{k}_{m,n}\}$, the constraints $\{\tmse^k_m\leq a^2\}$ can be equivalently presented as follows.
\begin{equation}\label{eq:tmse2}
    \left\{
    \begin{array}{ll}
      \sqrt{\sum_n(b^{k}_{m,n})^2+1}  \leq f^k_m a & \forall m,k\\
      \|\bh^k_{m,m}\bPhi_m-f^k_m\be_k\|\leq b^{k}_{m,m} &\forall m,k,\quad \forall \;\|\bDelta^k_{m,m}\|\leq\epsilon^k_{m,m},\\
      \|\bh^k_{m,n}\bPhi_n\|\leq b^{k}_{m,n} &\forall k,m\neq n\;\quad \forall \;\|\bDelta^k_{m,n}\|\leq\epsilon^k_{m,n},
    \end{array}\right.
\end{equation}
where we let $\be_k$ denote a length $K$ unit vector having a one in its $k^{th}$ position
 and zeros elsewhere.
Without loss of generality we have assumed $f^k_m\in\mathbb{R}^{+}$ as multiplying the vectors $\bw^k_m$ with any unit-magnitude complex scalar will not change the objective or the constraints of the problem $\bar\CS_2(\bP)$. Next, by applying the Schur Complement lemma   the constraints $\|\bh^k_{m,m}\bPhi_m-f^k_m\be_k\|\leq b^{k}_{m,m}$ for all $\|\bDelta^k_{m,m}\|\leq\epsilon^k_{m,m}$, can be equivalently stated as
\begin{equation*}
   \left[
    \begin{array}{ccc}
      b^{k}_{m,m} & (\tilde\bh^k_{m,m}+\bDelta^k_{m,m})\bPhi_m-f^k_m\be_k \\
      (\bPhi_m)^H(\tilde\bh^{k,k}_{m,m} + \bDelta^{k,k}_{m,m})^H-f^k_m\be_k^H & b^{k}_{m,m}\bI
    \end{array}\right]\succeq 0,\;\;\forall \; \|\bDelta^k_{m,m}\|\leq\epsilon^k_{m,m},
\end{equation*}
which upon using Lemma \ref{lemma:S} with $\bB\dff-[1\;\; \boldsymbol{0}]$, $\bC\dff[\boldsymbol{0}\;\;\bPhi_m]$, $\bD=\bDelta^k_{m,m}$, and
\begin{equation*}
   \bA\dff\left[
    \begin{array}{ccc}
      b^{k }_{m,m} & \tilde\bh^k_{m,m}\bPhi_m-f^k_m\be_k  \\
      (\bPhi_m)^H(\tilde\bh^k_{m,m})^H-f^k_m\be_k^H & b^{k}_{m,m}\bI
    \end{array}\right],
\end{equation*}
are  equivalently given by
\begin{equation}\label{eq:LMIx}
     \bT^k_m\dff\left[
    \begin{array}{cccc}
      b^{k }_{m,m}-\lambda^{k}_{m,m} & \tilde\bh^k_{m,m}\bPhi_m-f^k_m\be_k & \boldsymbol{0}\\
      (\bPhi_m)^H(\tilde\bh^k_{m,m})^H-f^k_m\be_k^H & b^{k}_{m,m}\bI & -\epsilon^k_{m,m}(\bPhi_m)^H\\
      \boldsymbol{0} & -\epsilon^k_{m,m}\bPhi_m  & \lambda^{k}_{m,m}\bI
    \end{array}\right]\succeq 0\quad\forall m,k.
\end{equation}
Similarly we can show that the constraints $\|\bh^k_{m,n}\bPhi_n\|\leq b^{k}_{m,n}$ holding for all $\|\bDelta^k_{m,n}\|\leq\epsilon^k_{m,n}$ are equivalently given by
\begin{equation}\label{eq:LMI2x}
     \bU^{k}_{m,n}\dff \left[
    \begin{array}{cccc}
      b^{k }_{m,n} - \lambda^{k}_{m,n}& \tilde\bh^k_{m,n}\bPhi_n& \boldsymbol{0}\\
      (\bPhi_n)^H(\tilde\bh^k_{m,n})^H & b^{k }_{m,n}\bI & -\epsilon^k_{m,n}(\bPhi_n)^H\\
      \boldsymbol{0} & -\epsilon^k_{m,n}\bPhi_n  & \lambda^{k }_{m,n}\bI
    \end{array}\right]\succeq 0\quad\forall m\neq n,k.
\end{equation}
Finally, note that the constraint $\sqrt{\sum_n(b^{k}_{m,n})^2+1}  \leq f^k_m a$ is equivalent to $\bV^{k}_{m} + f^k_m a\bI\succeq 0,\forall\; m,k$, where
\begin{equation}\label{eq:LMI22x}
     \bV^{k}_{m}\dff \left[
    \begin{array}{ccc}
      0 & \bb^{k}_m & 1 \\
      (\bb^k_m)^H & \boldsymbol{0} & 0\\
      1 & \boldsymbol{0}  & 0
    \end{array}\right],\forall\; m,k.
\end{equation}
Consequently, the problem $\bar\CS_2(\bP)$ is equivalent to
\begin{equation*}
    \left\{
    \begin{array}{ll}
      \min_{\{\bPhi_m,f^k_m\},\bb,\blambda,a} & a\\
      {\rm s.t.} & \bV^{k}_{m} + f^k_m a\bI\succeq 0\quad \forall m,k\\
      & \bT^k_m\succeq 0,\quad \forall k,m,\\
      & \bU^{k}_{m,n}\succeq 0,\quad\forall m\neq n,k \\
      & \|\bPhi_m\|^2_2\leq P_m\quad\forall\;m,
    \end{array}
    \right.
\end{equation*}
which is a standard form of GEVP \cite{Wiesel:SP06}.
\section{Proof of Theorem \ref{th:sum_rate}}
\label{app:th:sum_rate}
We first show that for any given and fixed  $\{f^k_m\}$, the problem $\bar\R(\bP,\bu)$ is equivalent to an SDP. We define
  $q^k_m=\alpha^k_m\exp(u^k_m-1)/|f^k_m|^2$ and see that
{\footnotesize\begin{align}
\nonumber \sum_{m=1}^M&\sum_{k=1}^K\alpha^k_m\exp(u^k_m-1) \tmse^k_m=
\sum_{m=1}^M\sum_{k=1}^K \underbrace{\left(q_m^k|\bh^k_{m,m}\bw^k_m-f^k_m|^2+\sum_{l\neq k}q_m^l|\bh^l_{m,m}\bw^k_m|^2 +\sum_{n\neq m}\sum_lq_n^l|\bh^l_{n,m}\bw^k_m|^2 \right)}_{\dff g(\bw_m^k)}
\end{align}}
Clearly, the optimization of $\bar\R(\bP,\bu)$ now decouples into $M$ optimization
 problems of the form
{\footnotesize \begin{equation}\label{eq:appopP}
    \left\{
    \begin{array}{ll}
      \min_{\bPhi_m,b_m} & b_m \\
      {\rm s.t.}  &   \max_{\{\bDelta^l_{n,m}\}}\sum_{k=1}^K g(\bw_m^k)\leq b_m \\
      & \|\bPhi_m\|^2_2\leq P_m\quad\forall\;m.
    \end{array}
    \right.
\end{equation}}
 Using the techniques employed in the proofs of Theorems 4 and 5, we can verify that the constraints can be equivalently expressed as finitely many LMIs so that the optimization problem is equivalent to an SDP. Next, suppose
 $\{\Phi_m\}$ are arbitrarily fixed. Then  $\bar\R(\bP,\bu)$ reduces to
 \begin{equation}
      \min_{\{f^k_m\}}  \sum_{m=1}^M\sum_{k=1}^K\alpha^k_m\exp(u^k_m-1) \;\max_{\{\bDelta^k_{m,n}\}}\tmse^k_m
\end{equation}
The above optimization problem decouples into $KM$ smaller problems of the form
\begin{equation}\label{eq:appopf}
      \min_{f^k_m}  \;\max_{\{\bDelta^k_{m,n}\}}\tmse^k_m
\end{equation}
Substituting $g_m^k=1/f^k_m$ in (\ref{eq:appopf}), we can optimize instead over $g^k_m$ and the latter optimization  problem can be readily shown to be equivalent to an SDP by using the techniques provided in \cite{Vucic:SP09}.

\renewcommand\url{\begingroup\urlstyle{rm}\Url}
{\small\bibliographystyle{IEEEtran}
\bibliography{IEEEabrv,Robust_BF}}

\end{document}